\documentclass{scrartcl}

\usepackage{amsmath,amssymb,amsfonts}
\usepackage{amsthm} % wird für theoremstyle benötigt

\usepackage{mathrsfs}
\usepackage[round]{natbib}
\usepackage{enumerate}
\usepackage{placeins} %gibt Objekte direkt aus
\usepackage[pdftex]{graphicx}
\usepackage{subcaption}
\usepackage{theoremref} %Referenz mit Theorem-Angabe
\usepackage{nccfoots}

\newtheorem{theorem}{Theorem}[section]
\newtheorem{definition}{Definition}[section]
\newtheorem{assumption}{Assumption}
\newtheorem{lemma}{Lemma}[section]
\newtheorem{nono-lemma}{Lemma}[section]
\newtheoremstyle{remarks}{}{}{}{}{\bfseries}{.}{4pt}{} %verhindert, dass die Bemerkungen kursiv geschrieben werden
\theoremstyle{remarks}
\newtheorem{rremark}{Remark}[section]
\DeclareMathOperator{\E}{E}
\DeclareMathOperator{\Var}{Var}
\DeclareMathOperator{\Cov}{Cov}
\DeclareMathOperator{\p}{P}

\begin{document}

\title{Robust Discrimination between Long-Range Dependence and a Change in Mean }
\author{Carina Gerstenberger$^*$}
%\footnote{Fakult\"at f\"ur Mathematik, Ruhr-Universit\"at Bo\-chum, 44780 Bochum, Germany}}
\date{}

%\email{carina.gerstenberger@rub.de}

%\keywords{keywords}
\maketitle

\begin{abstract}
\footnotesize
In this paper we introduce a robust to outliers Wilcoxon change-point testing procedure, for distinguishing between short-range dependent time series with a change in mean at unknown time and stationary long-range dependent time series. We establish the asymptotic distribution of the test statistic under the null hypothesis for $L_1$ near epoch dependent processes and show its consistency under the alternative. The Wilcoxon-type testing procedure similarly as the CUSUM-type testing procedure (of Berkes I., Horv{\'a}th L., Kokoszka P. and Shao Q. 2006. Ann. Statist. 34:1140-1165), requires estimation of the location of a possible change-point, and then using pre- and post-break subsamples to discriminate between short and long-range dependence. A simulation study examines the empirical size and power of the Wilcoxon-type testing procedure in standard cases and with disturbances by outliers. It shows that in standard cases the Wilcoxon-type testing procedure behaves equally well as the CUSUM-type testing procedure but outperforms it in presence of outliers. We also apply both testing procedure to hydrologic data.\par 
\end{abstract}

\Footnotetext{*}{Fakult\"at f\"ur Mathematik, Ruhr-Universit\"at Bo\-chum, 
44780 Bochum, Germany}

\section{Introduction}
Since the pioneering work of \citet{Hurst.1951}, \citet{Mandelbrot.1968a} and \citet{Mandelbrot.1968b}, the phenomenon of long-range dependence or Hust effect has been observed in many data sets, e.g. in hydrology, geophysics and economics. A lively debate also rages over the observed Hurst effect is due to long-range dependence or nonstationarity. \citet{Bhattacharya.1983} showed that the Hurst effect detected by $R/S$ statistics can be explained not only by long-range dependence, but by presence of a deterministic trend in short-range dependent data. \citet{Giraitis.2001} showed that some modified $R/S$ statistics reject the hypothesis of short-range dependence for long-range dependence but also for short-range dependent data in presence of a trend or change-points. The phenomenon of spurious long-range dependence has also been discussed in many other papers, see e.g. \citet{Granger.2004}.\par 
A first attempt for distinguishing between long-range dependence and short-range dependence with a monotonic trend was made by \citet{Kuensch.1986}, who showed that the periodogram in these two cases behaves differently. A test allowing to distinguish between a stationary long-range dependent process and short-range dependent process with a change in mean was introduced by \citet{Berkes.2006} and is based on the CUSUM statistic 
\begin{equation}\label{Cusum-TS}
C_{m,n}(k) = \sum_{i=m}^{k}X_i-\frac{k-m+1}{n}\sum_{i=1}^{n}X_i, \qquad m\leq k\leq n.
\end{equation}
It is well known that the CUSUM statistic is sensitive to outliers since it sums up the observations. In this paper we introduce a new robust to outliers testing procedure, which is based on the Wilcoxon change-point test statistic
\begin{equation}\label{Wilcoxon-TS}
W_{m,n}(k) = \sum_{i=m}^{k}\sum_{j=k+1}^{n}(1_{\{X_i\leq X_j\}}-1/2), \qquad m\leq k\leq n.
\end{equation}
\citet{Dehling.2013, Dehling.2015} used this test statistic for testing for changes in the mean of long-range dependent and short-range dependent processes respectively. In both papers the simulation studies point out that the Wilcoxon test statistic (\ref{Wilcoxon-TS}) is more robust to outliers than the CUSUM statistic (\ref{Cusum-TS}). Recently, \citet{Gerstenberger.2016} showed that Wilcoxon-type change-point location estimator for a change in mean of short-range dependent data based on test statistic (\ref{Wilcoxon-TS}) is also robust against outliers.\par

The new Wilcoxon-type testing procedure suggested in this paper is based on the idea of \citet{Berkes.2006}. Firstly, given a sample $X_1,\ldots,X_n$, one estimates the location $\hat{k}$ of a possible change in mean. Then the test statistic is defined as the maximum of the Wilcoxon change-point statistic (\ref{Wilcoxon-TS}) applied to the subsamples $X_1,\ldots,X_{\hat{k}}$ and $X_{\hat{k}+1},\ldots,X_n$.\par

\subsection*{Wilcoxon-type testing procedure}

Assuming that sample $X_1,\ldots,X_n$ is given, we want to test the hypothesis\par\bigskip
\hspace*{15pt}
{\boldmath$H_0$}: $X_i=Y_i+\mu_i$, $i=1,\ldots,n$ is generated by a stationary zero mean short-range dependent process $(Y_j)$ and has a change in mean $\mu_1=\ldots=\mu_{k^*}\neq \mu_{k^*+1}=\ldots=\mu_n$ at unknown time $k^*$,\par\bigskip

against the alternative\par\bigskip
\hspace*{15pt}
{\boldmath$H_1$}: $X_1,\ldots,X_n$ is a sample from a stationary long-range dependent process.
\par\bigskip
Note that during the paper stationary means strictly stationary.
\par
\bigskip
To construct the test statistic, first, we estimate the location $k^*$ of a change-point by a Wilcoxon-type change-point location estimator 
\begin{equation}\label{bp_estimator_wilcoxon}
\hat{k} = \min\Big\{k:\max_{1\leq l < n}\big|W_{1,n}(l) \big|=\big| W_{1,n}(k) \big|\Big\},
\end{equation}
which is defined as the smallest $k$ for which $|W_{1,n}(k)|$ attains its maximum.\par 
Next we divide the sample $X_1,\ldots,X_n$ into subsamples $X_1,\ldots,X_{\hat{k}}$ and $X_{\hat{k}+1},\ldots,X_n$, and set
\begin{equation*}%\label{TS_Wilcoxon}
T(X_1,\ldots,X_n) = n^{-3/2}\max_{1\leq k \leq n}\big|W_{1,n}(k) \big|.
\end{equation*}
Then we compute $T(X_1,\ldots,X_{\hat{k}})$ and $T(X_{\hat{k}+1},\ldots,X_n)$, and denote
\begin{align}
T_{n,1}:= T(X_1,\ldots,X_{\hat{k}}) &=\hat{k}^{-3/2}\max_{1\leq k \leq \hat{k}}\big|W_{1,\hat{k}}(k) \big|,\label{T_n1}\\
T_{n,2}:= T(X_{\hat{k}+1},\ldots,X_{n})&=(n-\hat{k})^{-3/2}\max_{\hat{k}< k \leq n}\big| W_{\hat{k}+1,n}(k) \big|.\label{T_n2}
\end{align}
Finally, we define the test statistic
\begin{equation}\label{WTP}
M_n = \max\{T_{n,1},T_{n,2}\}.
\end{equation}
\par
We show that $T(X_1,\ldots,X_n)$ allows to discriminate whether the sample has been generated by a short or long-range dependent stationary process. Hence, if we split the sample at time $\hat{k}$, which is close to the true change-point $k^*$, since $\hat{k}/k^*\rightarrow_p 1$ asymptotically we can assume that $X_1,\ldots,X_{\hat{k}}$ and $X_{\hat{k}+1},\ldots,X_n$ are samples from a stationary sequence with a constant mean. Subsequently, $M_n$ can be used to test if the samples $X_1,\ldots,X_{\hat{k}}$ and $X_{\hat{k}+1},\ldots,X_n$ have been generated by a short-range or long-range dependent stationary process.\par \bigskip

The outline of the paper is as follows. Section \ref{Assumptions and Main Results} specifies assumptions allowing to establish asymptotic distribution of $M_n$ under $H_0$ and consistency under $H_1$. Section \ref{Simulation Results} compares finite sample performance of the Wilcoxon-type and the CUSUM-type testing procedure. An application to hydrologic data is given in Section \ref{Data Example}. All proofs are given in Section \ref{Proofs}.

\section{Definitions, assumptions and main results}\label{Assumptions and Main Results}
In this section we present main assumptions, definitions and main results. \par 
Throughout the paper, $C$ denotes a generic non-negative constant, which may vary from time to time. 
The notation $a_n \sim b_n$ means that sequences $a_n$ and $b_n$ of real numbers have property $a_n/b_n\rightarrow c$, as $n\rightarrow \infty$, where $c\neq 0$. $\xrightarrow{d}$ and $\rightarrow_p$ stand for convergence in distribution and probability, respectively. By $\overset{d}{=}$ we denote equality in distribution. $\|g\|_{\infty}=\sup_{x}|g(x)|$ denotes the supremum norm of a function $g$.

\subsection*{Null hypothesis: short-range dependence with a change in mean}

Under the null hypothesis we assume the random variables $X_1,\ldots,X_n$ follow the change-point model
\begin{equation}\label{model_srd}
X_i = \begin{cases}
Y_i + \mu &, 1\leq i \leq k^*\\
Y_i + \mu + \Delta_n &, k^* < i \leq n,
\end{cases}
\end{equation}
where $k^*$ denotes the unknown location of the change-point in the mean, $\Delta_n$ denotes the unknown magnitude of change (see Assumption \ref{H2}) and $(Y_j)$ is a zero-mean strictly stationary short-range dependent process.\par
\bigskip

To cover a wide range of processes, we assume that the underlying process $(Y_j)$ can be written as $Y_j=f(Z_j,Z_{j-1},Z_{j-2},\ldots)$, $j\in \mathbb{Z}$, where $f:\mathbb{R}^{\mathbb{Z}}\rightarrow\mathbb{R}$ is a measurable function, and $(Z_j)$ is an absolutely regular (weakly dependent) process.

\begin{definition}
A stationary process $(Z_j)$ is called absolutely regular (or $\beta$-mixing) if
\begin{equation}\label{beta_k}
\beta_k = \sup_{n\geq 1}\E \sup_{A\in \mathcal{G}_{-\infty}^n} \left| \p\left(A|\mathcal{G}_{n+k}^{\infty}\right)-\p\left(A\right)\right| \rightarrow 0,
\end{equation}
as $k\rightarrow\infty$, where $\mathcal{G}_k^m$ is the $\sigma$-field generated by random variables $Z_k,\ldots,Z_m$, $k<m$.
\end{definition}
Absolute regularity or $\beta$-mixing implies the weaker property of $\alpha$-mixing, see e.g. \citet{Bradley.2002}.\par 

In addition, we will assume that $(Y_j)$ satisfies near epoch dependence condition, i.e. $Y_j$ depends on the near past of $(Z_j)$.

\begin{definition}
A stationary process $(Y_j)$ is $L_1$ near epoch dependent ($L_1$ NED) on some stationary process $(Z_j)$ with approximation constants $a_k$, $k\geq 0$, if
\begin{equation}\label{ned_condition}
\E|Y_1-\E(Y_1|\mathcal{G}_{-k}^{k})|\leq a_k, \qquad k=0,1,2,\ldots
\end{equation}
where $\mathcal{G}_{-k}^k$ is the $\sigma$-field generated by random variables $Z_{-k},\ldots,Z_k$ and $a_k\rightarrow 0$ as $k\rightarrow\infty$.
\end{definition}

Notice that a linear process or AR process might not be absolutely regular, but it would be $L_1$ near epoch dependent; see Example 2.1 in \citet{Gerstenberger.2016} for linear processes and \citet{Hansen.1991} for GARCH(1,1) processes. More examples of $L_1$ NED processes can be found in \citet{Borovkova.2001}, who also discuss more general $L_r$ NED processes, $r\geq 1$. The concept of $L_1$ near epoch dependence only assumes existence of the first moment $\E|Y_1|$. Therefore, we can allow heavy-tailed distributions.
\par \bigskip

We need further additional assumptions on the distribution function $F$ of $Y_1$,  the mixing coefficients $\beta_k$ in (\ref{beta_k}) and $a_k$ in (\ref{ned_condition}).

\begin{assumption}\label{H1}
The process $(Y_j)$ in (\ref{model_srd}) is $L_1$ NED on some absolutely regular process $(Z_j)$ with mixing coefficients $\beta_k$ and approximation constants $a_k$ such that
\begin{align}\label{cond_a_beta}
\sum_{k=1}^{\infty}k^2(\beta_k+\sqrt{a_k})<\infty.
\end{align}
Moreover, $Y_1$ has a continuous distribution function $F$ with bounded second derivative, and variables $Y_1-Y_k$, $k\geq 1$ satisfy
\begin{equation}\label{Bed_gem_Verteilung}
\p(x\leq Y_1-Y_k\leq y)\leq C|y-x|,
\end{equation}
for all $x\leq y$, where $C$ does not depend on $k$ and $x,y$. 
\end{assumption}
\par\bigskip

We suppose that both, the unknown change-point $k^*$ and the magnitude of change $\Delta_n$ in (\ref{model_srd}), depend on the sample size $n$.

\begin{assumption}\label{H2}
\begin{enumerate}[a)]
\item The change-point $k^*=[n\theta]$, where $0<\theta<1$ is fixed, is proportional to the sample size $n$. 
\item The magnitude of change $\Delta_n$ in (\ref{model_srd}) depends on $n$, and is such that
\begin{equation*} \label{assumption_size_change_infty}
\Delta_n \rightarrow 0, \qquad n\Delta_n^2 \rightarrow \infty, \qquad n\rightarrow \infty.
\end{equation*}
\end{enumerate}
\end{assumption}

An important step of our testing procedure is the estimation of the location $k^*$ of the change-point in mean. \citet{Gerstenberger.2016} showed that under Assumptions \ref{H1} and \ref{H2} the Wilcoxon-type change-point location estimator $\hat{k}$ in (\ref{bp_estimator_wilcoxon}) is consistent,
\begin{equation}\label{rate_consistency_bp_estimator}
\Delta_n^2 \big|\hat{k}-k^*\big|=O_P(1), \qquad \text{as}\; n\rightarrow \infty.
\end{equation}

\subsection*{Alternative: long-range dependence}

Under alternative $H_1$, the sample $X_1,\ldots,X_n$ is generated by a stationary long-range dependent process:
\begin{equation}\label{model_lrd}
X_i = G(\xi_i) + \mu, \qquad i=1,\ldots,n,
\end{equation}
where $\mu$ is the unknown mean and $(\xi_j)$ is a stationary long memory Gaussian process with $\E(\xi_1)=0,$ $\Var(\xi_1)=1$ and (non-summable) auto-covariances $\gamma_k=\Cov(\xi_1,\xi_{1+k})\sim k^{2d-1}c_0$, where $c_0>0$ and $d\in(0,1/2)$. Furthermore, we assume that $G:\mathbb{R}\rightarrow\mathbb{R}$ is a measurable, strictly monotone function such that $\E(G(\xi_1))=0$.

\subsection*{Main results}

The following theorem derives the limit distribution of the test procedure under the null hypothesis $H_0$. Below, $B(t) = W(t)-tW(1)$ denotes a standard Brownian bridge, where $W(t)$ is a standard Brownian motion.

\begin{theorem}\thlabel{theorem_hypothesis}
Let $\left(X_j\right)$ follow the model in (\ref{model_srd}). Then, under Assumptions \ref{H1} and \ref{H2},
\begin{equation}\label{claim_th}
 M_n = \max\{T_{n,1},T_{n,2}\} \xrightarrow{d}
\sigma\max\Big\{\sup_{0\leq t \leq 1}\big|B^{(1)}\left(t\right)\big|,\sup_{0\leq t \leq 1}\big|B^{(2)}\left(t\right)\big|\Big\} =:\sigma Z
\end{equation}
where $B^{(1)}$ and $B^{(2)}$ are two independent Brownian bridges,
\begin{equation}\label{lrv_sigma2}
\sigma^2 = \sum_{k=-\infty}^{\infty}\Cov\left(F(Y_0),F(Y_k)\right),
\end{equation}
and $F$ denotes the distribution function of $Y_1$.
\end{theorem}

Since the limit distribution of $M_n$ depends on the long-run variance $\sigma^2$, to calculate the critical values for the test, we need to estimate the long-run variance; see Section \ref{Simulation Results}. \par
\bigskip

We will compare performance of our test with the CUSUM-type test by \citet{Berkes.2006} defined as
\begin{equation}\label{CTP}
\tilde{M}_{C,n} =  \max\{\tilde{T}_C(X_1,\ldots,X_{\tilde{k}_C}),\tilde{T}_C(X_{\tilde{k}_C+1},\ldots,X_n)\},
\end{equation}
where
\[
\tilde{T}_C(X_1,\ldots,X_n) = (\hat{s}_n\sqrt{n})^{-1}\max_{1\leq k \leq n}\big|C_{1,n}(k)\big|,
\]
is based on the CUSUM statistic $C_{1,n}(k)$ in (\ref{Cusum-TS}). $
\tilde{k}_C  = \min\Big\{ k: \max_{1\leq l \leq n}\big| C_{1,n}(l) \big| =  \big| C_{1,n}(k) \big|  \Big\}$
is a CUSUM-type estimator of $k^*$ and $\hat{s}_n^2$ is a long-run variance estimator of $\sigma_c^2 = \sum_{k=-\infty}^{\infty}\Cov\left(Y_0,Y_k\right)$ given in (\ref{bartlett}).
 \citet{Berkes.2006} showed that under their assumptions under the null hypothesis, $\tilde{M}_{C,n}\xrightarrow{d} Z$.
\par\bigskip

The next theorem establishes consistency of the test $M_n$, i.e. that the test will detect long-range dependence with probability tending to 1.

\begin{theorem}\thlabel{theorem_alternative}
Let $(X_j)$ be as in (\ref{model_lrd}). Then, as $n\rightarrow\infty$,
\[
M_n \rightarrow_p \infty.
\]
\end{theorem}

Under the alternative in (\ref{model_lrd}) we do not consider the long memory Gaussian process itself, but a function of it. This concept also allows non-Gaussianity. We restrict the result of \thref{theorem_alternative} to strictly monotone functions due to simplicity of the proof. But the result can also be expanded to more general functions $G(\cdot)$. In this case the dependence structure of $(G(\xi_i))$ is in general not clear. Proposition 1.2 of \citet{Rooch.2012} yields that under slight assumptions if $\gamma_k \sim c_0k^{2d-1}$, $c_0>0$, $d\in (0,1/2)$ then $\Cov(G(\xi_i),G(\xi_{i+k}))\sim (c_0/m!) k^{(2d-1)m}$, where $m$ is the Hermite rank of $G$ (see Section \ref{Proof_alternative} for more details about Hermite rank). Therefore, for $-1<(2d-1)m<0$, the process $(G(\xi_i))$ is still long range dependent.\par
\bigskip
Proofs of \thref{theorem_hypothesis} and \ref{theorem_alternative} are given in Section \ref{Proofs}.
 
\section{Simulation Study}\label{Simulation Results}

In this simulation study we compare the finite sample performance (size and power) of the Wilcoxon-type testing procedure $M_n$ in (\ref{WTP}) with the CUSUM-type testing procedure $\tilde{M}_{C,n}$ of \citet{Berkes.2006}, given in (\ref{CTP}).\par
\bigskip
\textbf{Simulation set up}\par
To calculate the \textit{empirical size} we generate the sample of random variables $X_1,\ldots,X_n$ using the change-point model
\begin{equation}\label{simu_model_srd}
X_i = \begin{cases}
Y_i + \mu &, 1\leq i \leq k^*\\
Y_i + \mu + \Delta &, k^* < i \leq n,
\end{cases}
\end{equation}
where $Y_i = \rho Y_{i-1} + \epsilon_i$ is an AR(1) process with $\rho=0.4$. The innovations $\epsilon_i$ are generated from a standard normal distribution and a Student's t-distribution with $\nu=1$ degree of freedom. We set $k^*=[n\theta]$, $\theta = 0.25,0.5,0.75$ and $\Delta = 0.5,1,2$. \par
 Note that $t_1$-distributed innovations do not satisfy the $L_1$ NED condition, since $L_1$ NED requires the existence of $\E|Y_1|$. However, $t_1$-distributed innovations are included in the simulation study, since it proofs the functionality of Wilcoxon-type testing procedure even in the case of extremely heavy tails.\par 
 
To evaluate the \textit{empirical power} of the test we generate a sample $X_1,\ldots,X_n$ of fractional Gaussian noise (fGn)
\begin{equation}\label{simu_model_lrd}
X_i = W_H(i+1) - W_H(i),
\end{equation}
where $W_H(t)$, $H=d+1/2\in(1/2,1)$ is a fractional Brownian motion, see e.g. \citet{Mandelbrot.1968a}. The sequence $(X_j)$ is a long-range dependent process: $\Cov(X_1,X_{1+k})\sim k^{2d-1}c_0$ with long-range dependence parameter $d\in(0,1/2)$. We consider $d=0.1,0.2,0.3,0.4$. \par 

To analyse the robustness of Wilcoxon and CUSUM testing procedures to \textit{outliers}, we replace observations $X_{[0.2n]}, X_{[0.4n]}, X_{[0.6n]}, X_{[0.8n]}$ in the sample $(X_1,\ldots,X_n)$ (under the null hypothesis or alternative) by outliers $50X_{[0.2n]}, 50X_{[0.4n]}, 50X_{[0.6n]}$ and $50X_{[0.8n]}$.\par 
We consider sample sizes $n=200,500,1000,2000,5000$. All simulation results are based on $10,000$ replications.\par \bigskip

\textbf{Critical values}\par \bigskip

To analyse the empirical size and power, we need to know the critical values for the tests $M_n$ and $\tilde{M}_{C,n}$.\par 
By \thref{theorem_hypothesis}, under the null hypothesis, 
\begin{equation*}
M_n = \max\big\{T_{n,1},T_{n,2}\big\}\xrightarrow{d} \sigma Z.
\end{equation*}
Hence, if $\hat{\sigma}^2(X_1,\ldots,X_k)$ is a consistent estimator for the long-run variance $\sigma^2$ based on the sample $X_1,\ldots,X_k$, then 
\[
\hat{M}_n=\max\Big\{\frac{T_{n,1}}{\hat{\sigma}(X_1,\ldots,X_{\hat{k}})},\frac{T_{n,2}}{\hat{\sigma}(X_{\hat{k}+1},\ldots,X_n)}\Big\}\xrightarrow{d} Z.
\]
The same asymptotics holds for the CUSUM test: $\tilde{M}_{C,n}\xrightarrow{d} Z$, see Corollary 2.1 of \citet{Berkes.2006}. Thus, the critical value $c_{\alpha}$ for a given significance level $\alpha$ is obtained by solving
\begin{equation}\label{cv}
\p\big(Z>c_{\alpha}\big) = \alpha.
\end{equation}
 Since $B^{(1)}$ and $B^{(2)}$ are independent Brownian bridges, (\ref{cv}) reduces to
\begin{equation}\label{critical_values}
\p\Big(\sup_{0\leq t \leq 1}\big|B^{(1)}(t)\big|\leq c_{\alpha}\Big) = (1-\alpha)^{1/2},
\end{equation}
where $\sup_{0\leq t \leq 1}\big|B^{(1)}(t)\big|$ has the well-known Kolmogorov-Smirnov distribution, and its quantiles can be found in statistical tables. For $\alpha=5\%$ (\ref{critical_values}) implies $c_{5\%}=1.478$.\par 
\bigskip

\textbf{Estimation of long-run variance}\par \bigskip

The selection of a long-run variance estimate $\hat{\sigma}$ in $\hat{M}_n$  has a strong impact on the size and power properties of the tests in finite samples.\par \bigskip

To estimate the long-run variance $\sigma_c^2 = \sum_{k=-\infty}^{\infty}\Cov\left(Y_0,Y_k\right)$ in $\tilde{M}_{C,n}$ in (\ref{CTP}), \citet{Berkes.2006} suggested to use the Bartlett estimator 
\begin{equation}\label{bartlett}
\hat{s}_n^2 = \frac{1}{n}\sum_{i=1}^{n}\left(X_i-\bar{X}_n\right)^2+2\sum_{j=1}^{q\left(n\right)}\Big(1-\frac{j}{q+1}\Big)\frac{1}{n}\sum_{i=1}^{n-j}\left(X_i-\bar{X}_n\right)\left(X_{i+j}-\bar{X}_n\right),
\end{equation}
where $\bar{X}_n=n^{-1}\sum_{i=1}^nX_i$, with the bandwidth $q\left(n\right) = C\log_{10}\left(n\right)$. Table \ref{Table_CB} reports the empirical size (for $\theta=0.5$, $\Delta=1$) and power (for $d=0.4$) in $\%$ at significance level $5\%$ of $\tilde{M}_{C,n}$ test, with $\hat{s}_n^2$ as in (\ref{bartlett}) computed with bandwidth $15\log_{10}\left(n\right)$.  It shows that $\tilde{M}_{C,n}$ with Bartlett estimator $\hat{s}_n^2$ is too conservative and has low power against the alternative, which has also been pointed out by \citet{Baek.2012} and \citet{Dette.2013}.\par 

\begin{table}[h]
\begin{tabular}{c|c|c|c|c}
n = & 500 & 1000 & 2000 & 5000\\\hline
emp. size & 0.05 & 0.87 & 2.48 & 3.79\\\hline
power & 0.30 & 7.62 & 27.44 & 60.51
\end{tabular}
\caption{Empirical size and power of $\tilde{M}_{C,n}$ test using the Bartlett estimator.}\label{Table_CB}
\end{table}

In our simulation study to improve the performance of $\tilde{M}_{C,n}$ test we proceed as follows.
To estimate $\sigma^2_C$, instead of $\hat{s}_n^2$, we use the non-overlapping subsampling estimator of $\sigma_C^2$ by \citet{Carlstein.1986}, with block length $l_n$,
\begin{equation}\label{Carlstein_estimator}
\hat{\sigma}^2_C = \frac{1}{\left[n/l_n\right]}\sum_{i=1}^{\left[n/l_n\right]} \frac{1}{l_n}\bigg( \sum_{j=\left(i-1\right)l_n+1}^{il_n} X_j - \frac{l_n}{n}\sum_{j=1}^{n}X_j \bigg)^2,
\end{equation}
which yields better size and power balance for $\tilde{M}_{C,n}$, as seen from Tables \ref{Table_empirical_size_both_without} and \ref{Table_empirical_power}. This estimator has also been used by \citet{Dehling.2015} for a CUSUM-type test for changes in the mean of a short-range dependent process.\par 
In turn, for our test $\hat{M}_n$ to estimate $\sigma$ we shall use the Carlstein type estimator for long-run variance proposed by \citet{Dehling.2013b},
\begin{equation}\label{estimator_lrv}
\hat{\sigma}_W= \frac{1}{\left[n/l_n\right]}\sqrt{\frac{\pi}{2}}\sum_{i=1}^{\left[n/l_n\right]}\frac{1}{\sqrt{l_n}}\bigg|\sum_{j=(i-1)l_n+1}^{il_n}F_n(X_j)-\frac{l_n}{n}\sum_{j=1}^{n}F_n(X_j) \bigg|,
\end{equation}
where $F_n\left(x\right) = n^{-1}\sum_{i=1}^{n}1_{\{X_i\leq x\}}$. Note that $\hat{\sigma}_W$ estimates $\sigma$, not $\sigma^2$.\par 

The Carlstein estimator $\hat{\sigma}^2_C$ as well as the estimator $\hat{\sigma}_W$ (\ref{estimator_lrv}) are subsampling type estimators and require to choose a suitable block length $l_n$. The choice of $l_n$ is widely discussed in the literature. For AR(1)-processes \citet{Carlstein.1986} suggests to use
\begin{equation}\label{block_length_carlstein}
l_n = \max\big\{\big\lceil n^{1/3}(2\rho/(1-\rho^2))^{2/3}\big\rceil,1\big\},
\end{equation}
where $\rho$ denotes the autocorrelation coefficient at lag 1. In our simulation study we use this block length with $\rho$ estimated by the sample autocorrelation coefficient $\hat{\rho}$ since it yields good results for the empirical size and power.\par 
In the presence of outliers, we need to robustify further the choice of the block length. Since the sample autocorrelation is highly sensitive to outliers, we use in (\ref{block_length_carlstein}) a robust estimator of $\rho$ proposed by \citet{Ma.2000}, 
\begin{equation*}
\hat{\rho}_Q = \frac{Q_{n-1}^2(u+v)-Q_{n-1}^2(u-v)}{Q_{n-1}^2(u+v)+Q_{n-1}^2(u-v)},
\end{equation*}
where $Q_n(x)=2.21914\{|X_i-X_j|;i<j\}_{(k)}$, $x=(X_1,\ldots,X_{n})$, which is the $k={{n}\choose{2}}/4$-th order statistic of the ${n}\choose{2}$ interpoint distances, is a robust scale estimator introduced by \citet{Rousseeuw.1993}, $u = (X_1,\ldots,X_{n-1})$ and $v=(X_{2},\ldots,X_n)$. Figure \ref{histogram.rho} contains the histogram of estimates $\hat{\rho}$ and $\hat{\rho}_Q$ based on 10,000 replications of sample $X_1,\ldots,X_{500}$ with outliers, generated by an AR(1) model with $\rho=0.4$ and i.i.d. standard normal innovations. For a further discussion on robust estimation of autocorrelation function see \citet{Duerre.2015}.\par \bigskip

 \begin{figure}
  \begin{subfigure}[c]{0.49\textwidth}
\includegraphics[width=7cm, height=5cm]{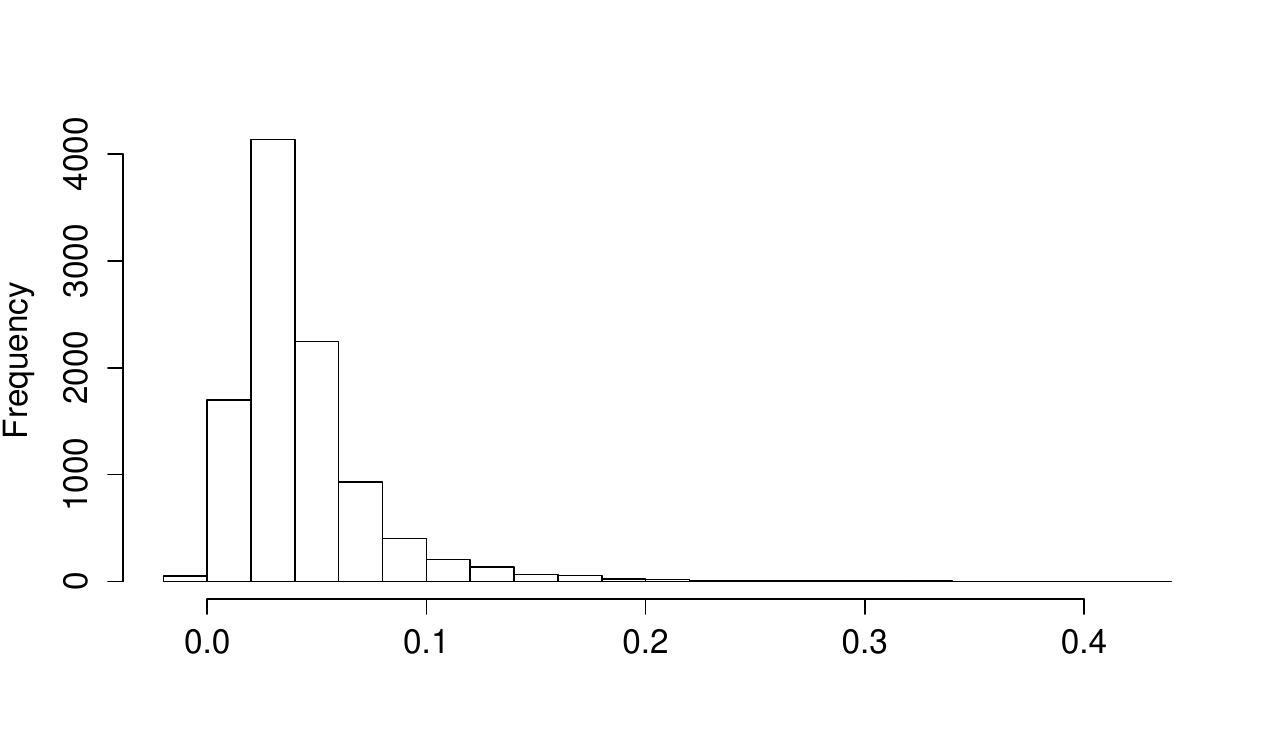}
\subcaption{sample autocorrelation}
\end{subfigure}
 \begin{subfigure}[c]{0.49\textwidth}
\includegraphics[width=7cm, height=5cm]{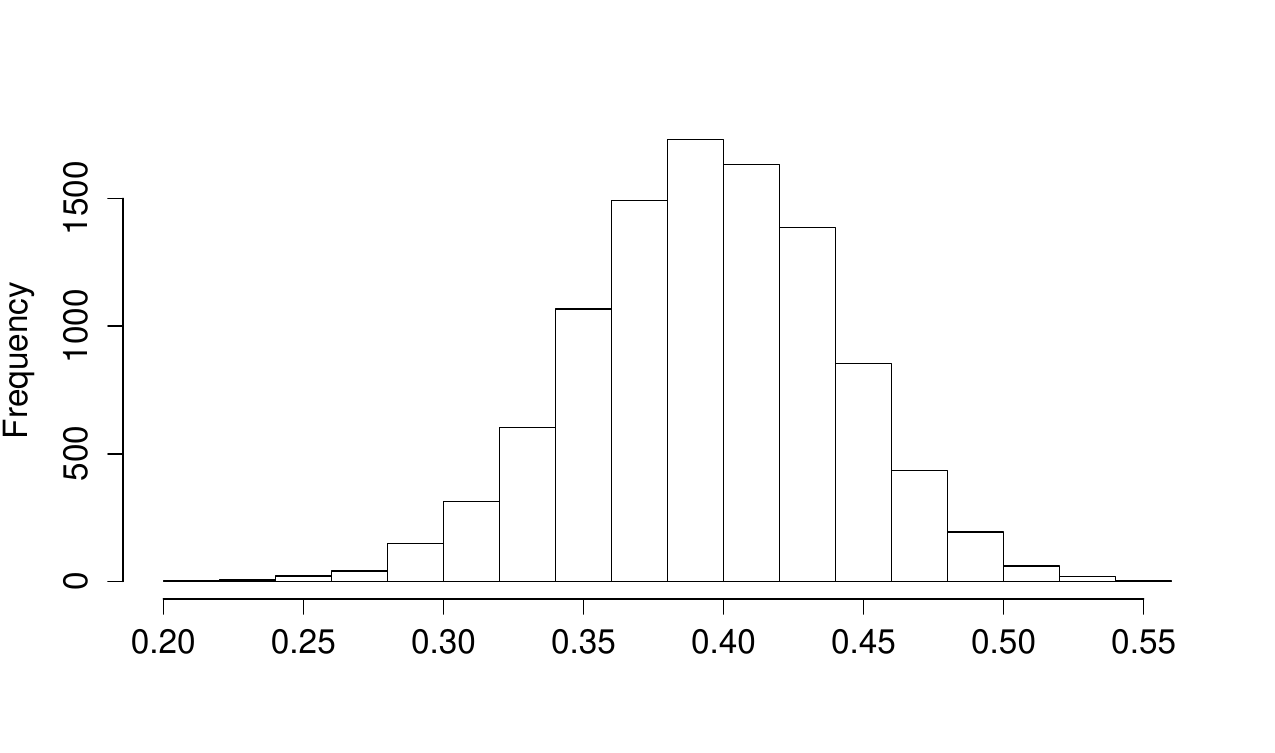}
\subcaption{MG-estimator}
\end{subfigure}
\caption{Histogram of $\hat{\rho}(1)$ and $\hat{\rho}_Q(1)$ based on 10,000 replications. $X_i$ is generated by an AR(1) process with outliers, $\epsilon_i\sim N(0,1)$, $\rho=0.4$ and $n=500$.}\label{histogram.rho}
\end{figure}

\textbf{Simulation results}\par \bigskip

Table \ref{Table_empirical_size_both_without} reports the empirical size at the $5\%$ significance level based on 10,000 replications of $\tilde{M}_{C,n}$ and $\hat{M}_n$ tests, for the model (\ref{simu_model_srd}) without outliers. The empirical size of $\hat{M}_n$ and $\tilde{M}_{C,n}$ slightly exceed the $5\%$ level for large sample size $n$ for $\theta=0.5$ and $\Delta=0.5,1,2$. The size of the tests is more distorted if the change-point is located close to the beginning or end of the sample, i.e. for $\theta=0.25,0.75$. We also consider the situation of no change, i.e. $\Delta=0$, for which the empirical size of both testing procedures is close to the nominal size. Empirical sizes of $\hat{M}_n$ and $\tilde{M}_{C,n}$ are comparable in the absence of outliers. \par 
Note that in Table \ref{Table_empirical_size_both_without} both tests do not tend to $5\%$ as it is expected. This is due to a very slow convergence to the limit process. In simulation studies with really large sample size $n>10,000$ the empirical size of both tests is tending to $5\%$. Since $\tilde{M}_{C,n}$ and $\hat{M}_n$ are both suffering from this slow convergence, they are still comparable to each other.
\par
Table \ref{Table_empirical_size_both_with} reports the empirical size of $\hat{M}_n$ and $\tilde{M}_{C,n}$ in presence of outliers and $t_1$-distributed innovations. While test $\hat{M}_n$ is robust to the outliers and just slightly affected by the heavy-tailed innovations, the test $\tilde{M}_{C,n}$ becomes much too conservative.
 
\begin{table}
\begin{tabular}{c||cc|cc|cc||cc}
$\theta=$ & \multicolumn{2}{|c|}{0.25} & \multicolumn{2}{|c|}{0.5} & \multicolumn{2}{|c||}{0.75} & \multicolumn{2}{|c}{0.5} \\\hline
 & $\tilde{M}_{C,n}$ & $\hat{M}_n$ & $\tilde{M}_{C,n}$ & $\hat{M}_n$ & $\tilde{M}_{C,n}$ & $\hat{M}_n$ & $\tilde{M}_{C,n}$ & $\hat{M}_n$\\\hline\hline
n= &  \multicolumn{6}{l||}{$\Delta = 1$} &  \multicolumn{2}{l}{$\Delta = 0$}\\\hline
$200$   &3.79 &3.52 &3.90 &3.41 &4.46 &3.92 &3.48 &2.78 \\
$500$   &8.35 &7.71 &5.12 &4.28 &8.47 &8.10 &4.36 &3.89 \\
$1000$  &9.83 &9.44 &5.11 &4.68 &10.10 &9.49 &4.61 &4.11  \\
$2000$  &9.45 &9.37 &5.96 &5.23 &9.87&9.76 &5.10 &4.64  \\
$5000$  &8.28 &7.77 &6.26 &5.59 &8.51 &8.01 &5.18 &4.91  \\\hline\hline
n= &  \multicolumn{6}{l||}{$\Delta = 2$} &  \multicolumn{2}{l}{$\Delta = 0.5$}\\\hline
$200$   &5.08 &4.68 &4.18 &3.69 &5.85 &5.12 &3.63 &3.03 \\
$500$   &7.32 &8.03 &5.49 &4.67 &7.07 &7.43 &4.54 &4.10 \\
$1000$  &7.67 &8.05 &5.38 &4.79 &7.15 &7.38 &4.82 &4.46  \\
$2000$  &7.11 &7.16 &6.03 &5.31 &6.88 &7.15 &5.57 &4.90  \\
$5000$  &6.30 &6.12 &6.15 &5.58 &6.45 &6.29 &6.01 &5.46  \\\hline\hline
\end{tabular}
\caption{Empirical size of $\tilde{M}_{C,n}$ and $\hat{M}_n$ tests at the $5\%$ significance level, 10,000 replications. $X_i$ follows the model (\ref{simu_model_srd}) without outliers and $\epsilon_i\sim N(0,1)$.}\label{Table_empirical_size_both_without}
\end{table}

\begin{table} 
\begin{tabular}{c|cc|cc|cl}
 & \multicolumn{2}{|c}{$\epsilon_i\sim N(0,1)$} & \multicolumn{2}{|c}{$\epsilon_i\sim t_1$}  &\multicolumn{2}{|c}{$\epsilon_i\sim N(0,1)$ with	outliers}  \\\hline
 n=  & \multicolumn{1}{|c}{$\tilde{M}_{C,n}$} & \multicolumn{1}{c}{$\hat{M}_n$} & \multicolumn{1}{|c}{$\tilde{M}_{C,n}$}  & \multicolumn{1}{c}{$\hat{M}_n$} & \multicolumn{1}{|c}{$\tilde{M}_{C,n}$} & \multicolumn{1}{l}{$\hat{M}_n$}\\\hline\hline
$1000$ 	&5.11 &4.68  &0.83 &2.92  &0.56 &4.82\\\hline
$2000$	&5.96 &5.23  &1.22 &3.74  &1.17 &5.56\\\hline
$5000$  &6.26  &5.59   &1.03 &4.57  &2.28 &5.41\\\hline
\end{tabular}	
\setcapindent{0pt}\caption{Empirical size of $\tilde{M}_{C,n}$ and $\hat{M}_n$ tests at the $5\%$ significance level, 10,000 replications.  $X_i$ follows the model (\ref{simu_model_srd}) with $\epsilon_i\sim N(0,1)$ without and with outliers, and $\epsilon_i\sim t_1$. We consider $\Delta=1$ and $\theta=0.5$.}\label{Table_empirical_size_both_with}
\end{table}

Tables \ref{Table_empirical_power} and  \ref{Table_empirical_power_outliers} report the empirical power of test $\tilde{M}_{C,n}$ and $\hat{M}_n$, for $X_i$ in (\ref{simu_model_lrd}) without outliers and with outliers, respectively. Table \ref{Table_empirical_power} shows that the power of both tests increases with increasing sample size and dependence parameter $d$ (except power of $\hat{M}_n$ for $n=200$, $d=0.4$). It shows that in absence of outliers $\hat{M}_n$ and $\tilde{M}_{C,n}$ have similar power properties.

\begin{table}
\begin{tabular}{c|cc|cc|cc|cc}
d =  & \multicolumn{2}{|c|}{0.1} & \multicolumn{2}{|c|}{0.2} & \multicolumn{2}{|c|}{0.3} & \multicolumn{2}{|c}{0.4} \\\hline
n=				& $\tilde{M}_{C,n}$	& $\hat{M}_n$ & $\tilde{M}_{C,n}$	& $\hat{M}_n$ & $\tilde{M}_{C,n}$	& $\hat{M}_n$ & $\tilde{M}_{C,n}$	& $\hat{M}_n$ \\\hline\hline
200		  		&7.68	&5.90 &12.28 &9.99 &14.11 &11.50 &12.53 &9.35\\\hline     
500				&14.12 &11.53 &25.31 &22.84 &31.52 &28.33 &32.03 &28.42\\\hline
1000			&20.22 &16.95 &35.37 &32.64	&46.41 &43.11 &50.22 &46.06 \\\hline
2000			&26.67 &23.90 &49.17 &45.95 &61.92 &58.68 &67.50 &63.52\\\hline
5000			&35.05 &32.68 &64.44 &61.27 & 79.67 &77.48 &85.12 &82.63\\\hline
\end{tabular}
\caption{Empirical power  of $\tilde{M}_{C,n}$ and $\hat{M}_n$ tests at the $5\%$ significance level, 10,000 replications. $X_i$ follows the model (\ref{simu_model_lrd}) without outliers.}\label{Table_empirical_power}
\end{table} 

Table \ref{Table_empirical_power_outliers} shows that the empirical size of $\hat{M}_n$ is practically not affected by the outliers, whereas $\tilde{M}_{C,n}$ suffers a loss of power.\par

Let us have a closer look on what happens in the case of outliers. There are different steps in the testing procedures that might be affected by outliers: the estimation of the time of change, the estimation of the long-run variance and the test statistic itself. The impact of outliers on a CUSUM and Wilcoxon based change-point estimator has already been discussed in \citet{Gerstenberger.2016}. It is shown that the Wilcoxon-type estimator is nearly not affected by outliers whereas the CUSUM-type estimator has trouble in detecting the correct time of change. Therefore, if this would be the only problem in the CUSUM-type testing procedure, we should expect $\tilde{M}_{C,n}$ to reject the hypothesis more often due to splitting the data at the spuriously estimated change-point. But as we have seen in Table \ref{Table_empirical_size_both_with} this is not the case. Let us now have a closer look at the CUSUM statistic $C_{1,n}(k)$ and the Wilcoxon statistic $W_{1,n}(k)$. We generated a series of random variables $Y_1, \ldots, Y_{n}$, $n=1000$ following the AR(1) process given in (\ref{simu_model_srd}), but without a change in mean. In Figure \ref{test_without_LRV} the solid line shows in (a) $n^{-1/2}|C_{1,n}(k)|$, $k=1,\ldots,1000$ and in (b) $n^{-3/2}|W_{1,n}(k)|$, $k=1,\ldots,1000$, both applied to $Y_1, \ldots, Y_{1000}$. Then we disturbed the same variables $Y_1,\ldots, Y_n$ with outliers as described above. The dashed lines in both figures show the results for $n^{-1/2}|C_{1,n}(k)|$ and $n^{-3/2}|W_{1,n}(k)|$ applied to the variables including outliers. We see again that the Wilcoxon statistic is not affected by the outliers. But as expected, the CUSUM statistic has larger values in the outlier scenario and therefore it has a larger maximum. But again, this should lead to a more often rejection of the hypothesis. So why do the simulation results show more conservatism for the CUSUM-type testing procedure in the outlier scenario? This is due to the long-run variance estimation. If we have a look at the value for the estimator given in (\ref{Carlstein_estimator}) applied to the example we see that the value for the data with outliers ($\hat{\sigma}^2_C = 4.63$) is much higher than the value for data without outliers ($\hat{\sigma}^2_C = 2.04$). This reduces the values for the CUSUM-testing procedure for outlier scenario, since we divide by the estimate of the long-run variance, see Figure \ref{test_with_LRV} (a). This leads to reduction of size and a loss in power. For the Wilcoxon-type testing procedure we can observe that the value of $\hat{\sigma}_W$ given in (\ref{estimator_lrv}) is in both cases nearly the same ($\hat{\sigma}_W=0.38$ with outliers and $\hat{\sigma}_W=0.41$ without), see Figure \ref{test_with_LRV} (b).
 \begin{figure}
  \begin{subfigure}[c]{0.49\textwidth}
\includegraphics[width=7cm, height=5cm]{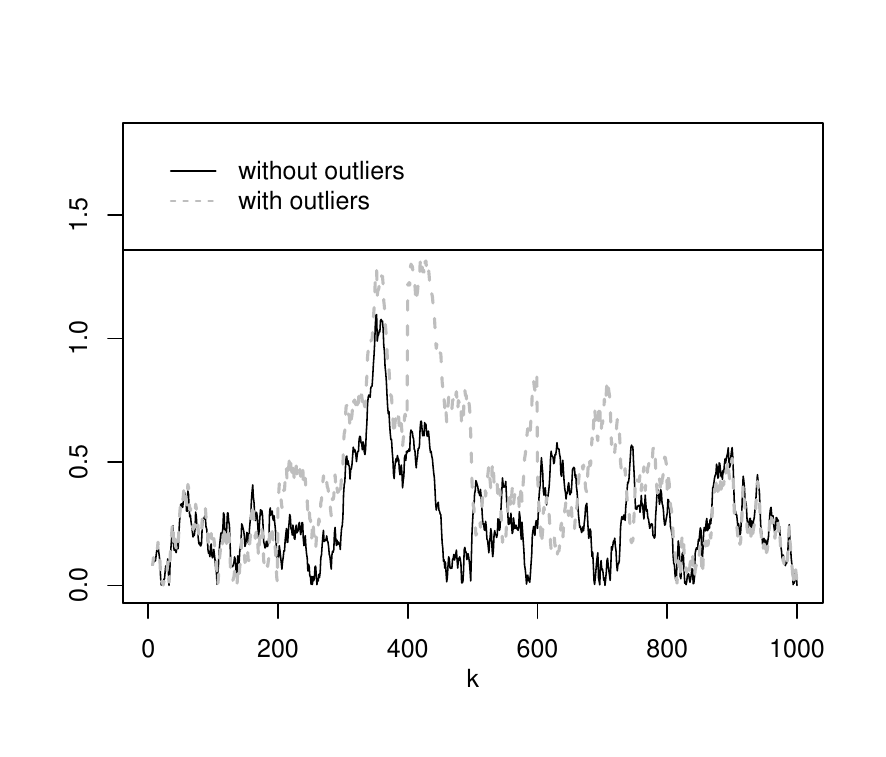}
\subcaption{$n^{-1/2}|C_{1,n}(k)|$}
\end{subfigure}
 \begin{subfigure}[c]{0.49\textwidth}
\includegraphics[width=7cm, height=5cm]{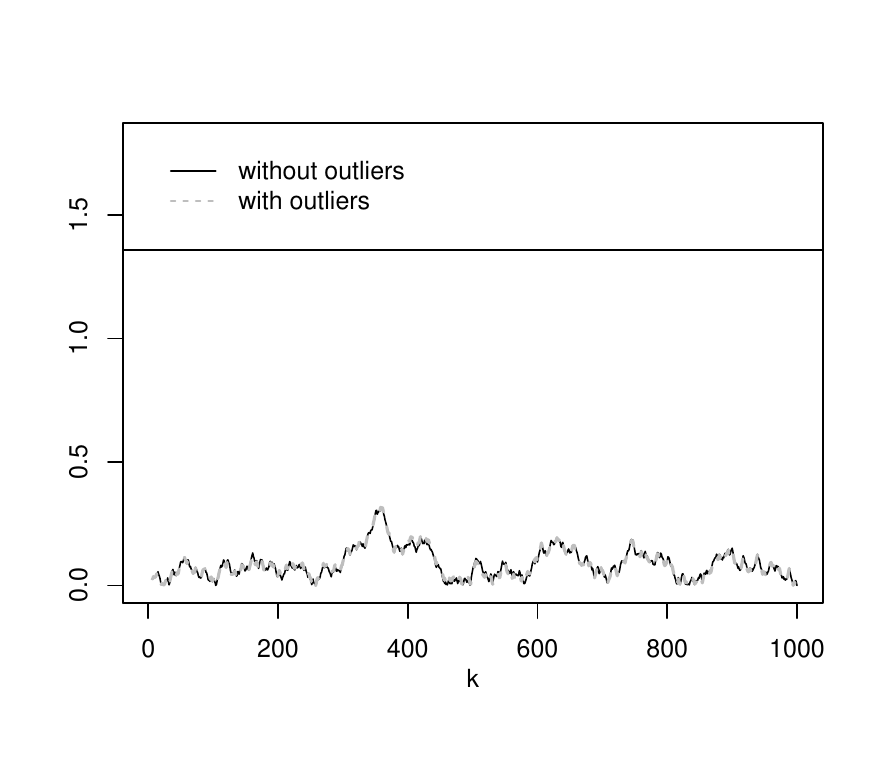}
\subcaption{$n^{-3/2}|W_{1,n}(k)|$}
\end{subfigure}
\caption{Values of $n^{-1/2}|C_{1,n}(k)|$ and $n^{-3/2}|W_{1,n}(k)|$ for $k=1,\ldots,1000$. $Y_i = \rho Y_{i-1} + \epsilon_i$ is an AR(1) process with $\rho=0.4$ and standard normal innovations $\epsilon_i$. For the dashed lines $(Y)$ is disturbed by outliers.}\label{test_without_LRV}
\end{figure}
 \begin{figure}
  \begin{subfigure}[c]{0.49\textwidth}
\includegraphics[width=7cm, height=5cm]{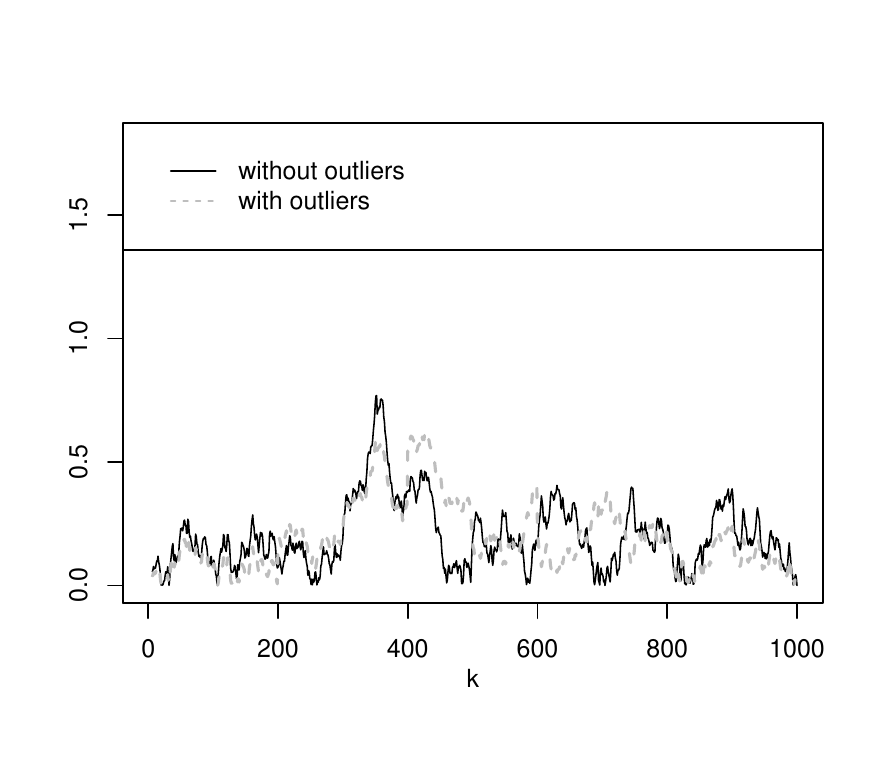}
\subcaption{$(\hat{\sigma}_C n^{1/2})^{-1}|C_{1,n}(k)|$}
\end{subfigure}
 \begin{subfigure}[c]{0.49\textwidth}
\includegraphics[width=7cm, height=5cm]{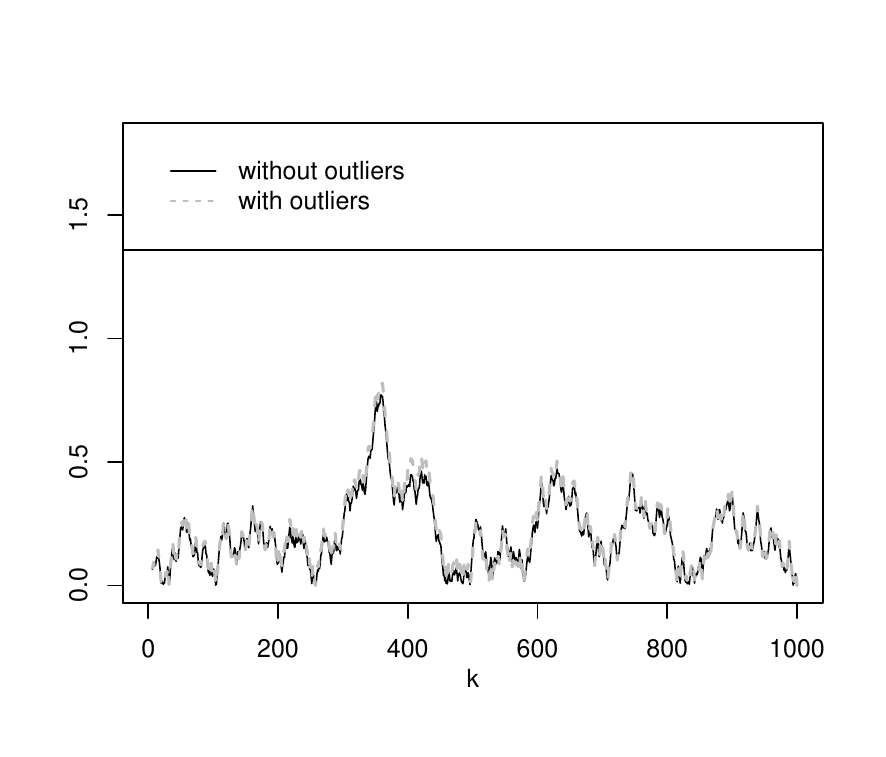}
\subcaption{$(\hat{\sigma}_W n^{3/2})^{-1}|W_{1,n}(k)|$}
\end{subfigure}
\caption{Values of $(\hat{\sigma}_C n^{1/2})^{-1}|C_{1,n}(k)|$ and $(\hat{\sigma}_W n^{3/2})^{-1}|W_{1,n}(k)|$ for $k=1,\ldots,1000$. $Y_i = \rho Y_{i-1} + \epsilon_i$ is an AR(1) process with $\rho=0.4$ and standard normal innovations $\epsilon_i$. For the dashed lines $(Y)$ is disturbed by outliers.}\label{test_with_LRV}
\end{figure}
\par 
\bigskip
In general, we conclude that Wilcoxon test $\hat{M}_n$ allows discrimination between long-range dependence and short-range dependence with a change in mean that is robust to outliers. In absence of outliers it performs equally well as CUSUM test $\tilde{M}_{C,n}$, but outperforms it in presence of outliers.

 \begin{table}
\begin{tabular}{c|cc|cc|cc|cc}
d =  & \multicolumn{2}{|c|}{0.1} & \multicolumn{2}{|c|}{0.2} & \multicolumn{2}{|c|}{0.3} & \multicolumn{2}{|c}{0.4} \\\hline
n=				& $\tilde{M}_{C,n}$	& $\hat{M}_n$ & $\tilde{M}_{C,n}$	& $\hat{M}_n$ & $\tilde{M}_{C,n}$	& $\hat{M}_n$ & $\tilde{M}_{C,n}$	& $\hat{M}_n$ \\\hline\hline
200		  		& 1.63  &6.06 &2.53 &10.06 &2.65 &11.88 &3.62 &9.69\\\hline     
500				&2.76 &11.71&5.02 &22.95&7.26 &28.60&8.69 &28.37\\\hline
1000			&4.10 &17.13&10.40 &32.60 &16.91 &43.11&21.96 &46.18\\\hline
2000			&8.46 &23.88&23.07 &45.90& 37.05 &58.71&47.00&63.68\\\hline
5000			&18.76 &32.66&46.78 &61.55 &68.99 &77.54&78.65 &82.68\\\hline
\end{tabular} 
\caption{Empirical power of $\tilde{M}_{C,n}$ and $\hat{M}_n$ tests at the $5\%$ significance level, 10,000 replications. $X_i$ follows the model (\ref{simu_model_lrd}) with outliers.}\label{Table_empirical_power_outliers}
\end{table}

%\FloatBarrier %zwingt die Grafiken in den obigen Kapitel zu bleiben

\section{Data Example}\label{Data Example}

In the following data example we consider a hydrologic time series. In particular, we consider the mean daily discharges (MQ) of the river Elbe in Dresden, Germany. The data cover the time from 01.01.1844 to 31.12.1849 ($n=2191$) and are shown in figure \ref{MQ_example} (a). It is well known that daily MQ are strongly correlated, see figure \ref{acf_MQ} for the sample autocorrelation function. Hence, testing for dependency should result in long-range dependence. In the year 1845 there was a big flood in Dresden, which appears in figure \ref{MQ_example} (a) as an outlier. The time series also contains some smaller outliers after 1845. \par 
We calculated the CUSUM testing procedure $\tilde{M}_{C,n}$ and the Wilcoxon testing procedure $\hat{M}_{n}$ for each time point $k=1,\ldots,2191$. That means we divide the sample at the estimated time of change $\hat{k}$ and consider $(\hat{\sigma}_C \hat{k}^{1/2})^{-1}|C_{1,\hat{k}}(k)|$ for $k=1,\ldots,\hat{k}$ and $(\hat{\sigma}_C (n-\hat{k})^{1/2})^{-1}|C_{\hat{k}+1,n}(k)|$ for $k = \hat{k}+1,\ldots,n$ for the CUSUM test and $(\hat{\sigma}_W \hat{k}^{3/2})^{-1}|W_{1,\hat{k}}(k)|$ and $(\hat{\sigma}_W (n-\hat{k})^{3/2})^{-1}|W_{\hat{k}+1,n}(k)|$, respectively, for the Wilcoxon test. The results are shown in figure \ref{MQ_example} (b). The vertical line in the plot refers to the critical value $c_{5\%}=1.478$. \par
Although the data seem to be long-range dependent both testing procedures have a maximum value less than the critical value, where the CUSUM test has a much smaller value $\tilde{M}_{C,n} = 0.89$ than the Wilcoxon test $\hat{M}_{n} = 1.30$. This seems to be in line with the conclusion of the simulation section that the CUSUM test loses power due to the affect of outliers on the long-run variance estimation. Even though the Wicoxon test would also not reject, the value is close to the critical value.

 \begin{figure}
  \begin{subfigure}[c]{0.49\textwidth}
\includegraphics[width=7cm, height=5.4cm]{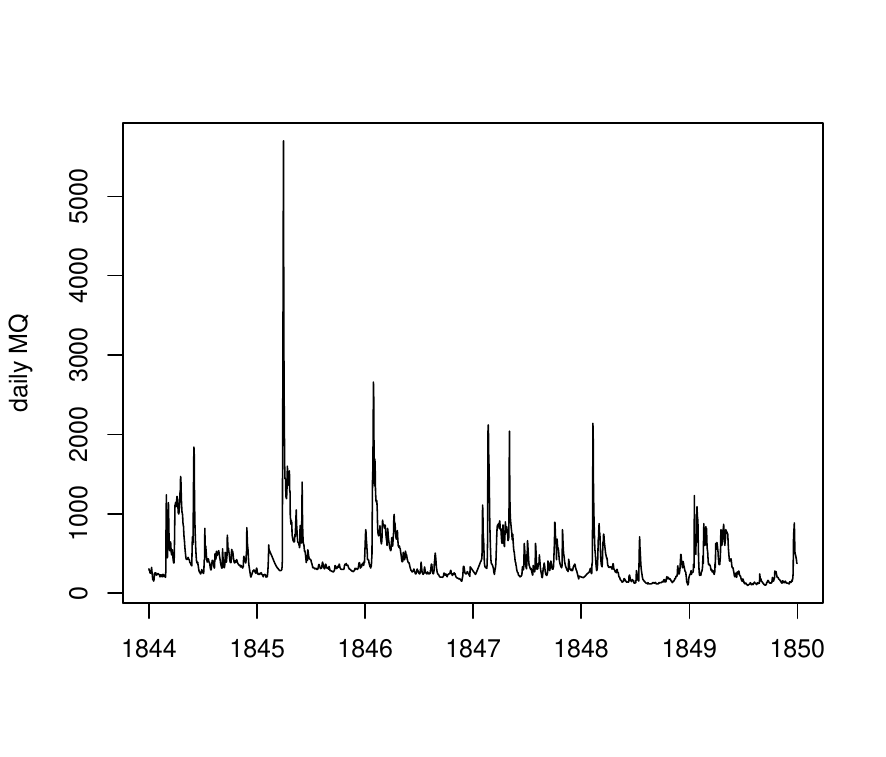}
\subcaption{Daily MQ of the river Elbe in Dresden, Germany.}
\end{subfigure}
 \begin{subfigure}[c]{0.49\textwidth}
\includegraphics[width=7cm, height=5.4cm]{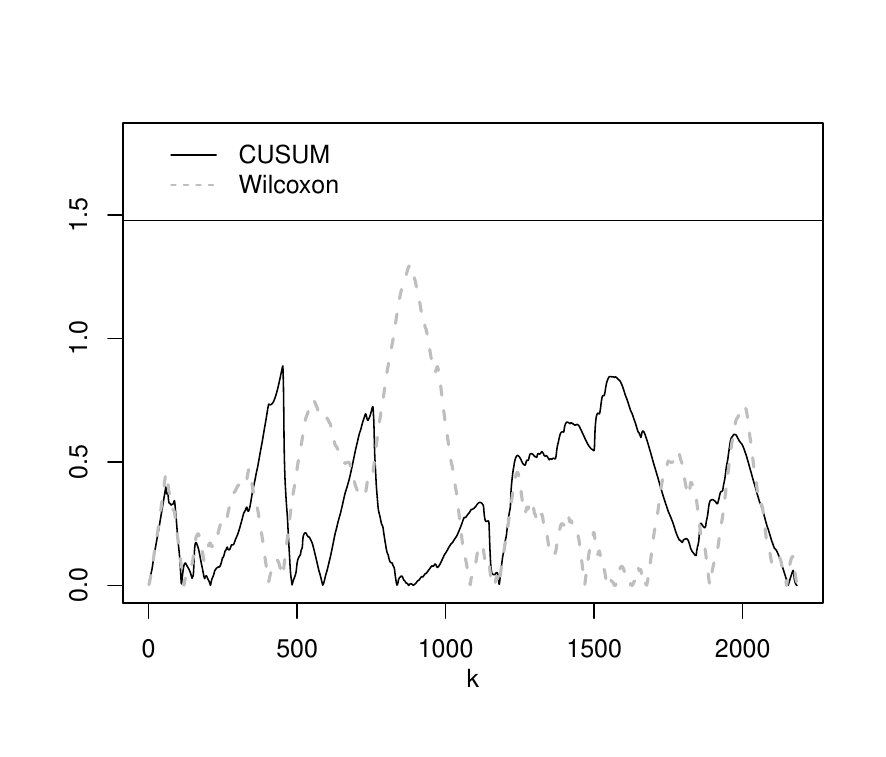}
\subcaption{Values for CUSUM and Wilcoxon testing procedures.}
\end{subfigure}
\caption{Mean daily discharge (MQ) of the river Elbe in Dresden, Germany, from 1844 to 1849 (a). In (b) we see the corresponding pointwise values for the CUSUM and Wilxocon type testing procedure.}\label{MQ_example}
\end{figure}

 \begin{figure}
\includegraphics[width=7cm, height=5cm]{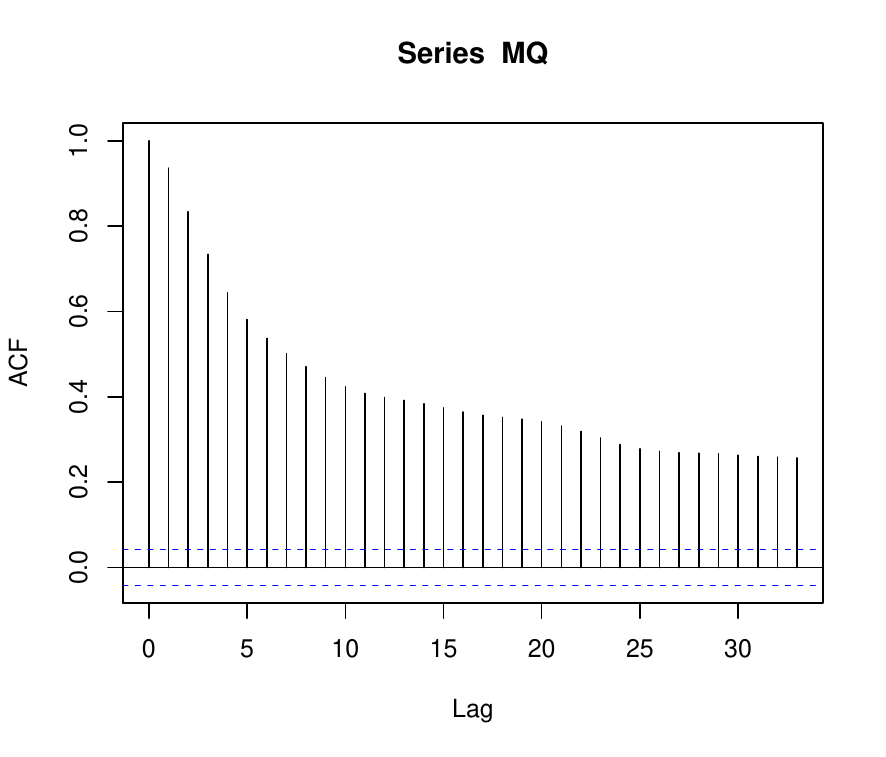}
\caption{Sample autocorrelation function of the daily MQ of the river Elbe in Dresden.}\label{acf_MQ}
\end{figure}
%\FloatBarrier %zwingt die Grafiken in den obigen Kapitel zu bleiben

\section{Proofs}\label{Proofs}

This section contains the proofs of \thref{theorem_hypothesis}, \thref{theorem_alternative} and auxiliary lemmas.

\subsection{Proof of \thref{theorem_hypothesis}}\label{Proof_hypothesis}

Suppose that $X_1,\ldots,X_n$ follow the model in (\ref{model_srd}) and Assumptions \ref{H1} and \ref{H2} are satisfied. Throughout the proofs without loss of generality, we assume $\mu = 0$ and $\Delta_n > 0$.\\

Before we can state the proof of \thref{theorem_hypothesis}, we need to consider the following lemmata, which proofs can be found in sections \ref{sec_lemma_1} and \ref{sec_lemma_2}, respectively.

\begin{lemma}\thlabel{lemma1}
Let $X_1,\ldots,X_n$ follow the model in (\ref{model_srd}), and Assumptions \ref{H1} and \ref{H2} be satisfied. Let $\hat{k}$ be defined as in (\ref{bp_estimator_wilcoxon}). Then,
\begin{align*}
n^{-3/2}\max_{1\leq k \leq \hat{k}} \big|W_{1,\hat{k}}(k) \big| 
&= n^{-3/2}\max_{1\leq k \leq \hat{k}} \big| \sum_{i=1}^{k}\sum_{j=k+1}^{\hat{k}}(1_{\left\{Y_i\leq Y_j\right\}}-1/2) \big| + o_P\left(1\right)\\
n^{-3/2}\max_{\hat{k}< k \leq n} \big| W_{\hat{k}+1,n}(k) \big|
&= n^{-3/2}\max_{\hat{k}< k \leq n} \big| \sum_{i=\hat{k}+1}^{k}\sum_{j=k+1}^{n}(1_{\left\{Y_i\leq Y_j\right\}}-1/2) \big| + o_P\left(1\right).
\end{align*}
\end{lemma}\bigskip

\begin{lemma}\thlabel{lemma2}
Let $\left(Y_j\right)$ satisfy Assumption \ref{H1} and let Assumption \ref{H2} hold. Then,
\begin{equation}\label{eq_lemma2}
\Big(T(Y_1,\ldots,Y_{\hat{k}}),T(Y_{\hat{k}+1},\ldots,Y_{n})\Big)\xrightarrow{d} \Big( \sigma \sup_{0\leq t \leq 1}\big| B^{(1)}\left(t\right) \big|, \sigma \sup_{0\leq t \leq 1}\big| B^{(2)}\left(t\right) \big|\Big),
\end{equation}
where $B^{(1)}$ and $B^{(2)}$ are independent Brownian bridges, and $\sigma$ is given in (\ref{lrv_sigma2}).
\end{lemma}\bigskip

\begin{proof}[Proof of \thref{theorem_hypothesis}]
We divide the proof  into two steps, as in the proof of Theorem 2.1 in \citet{Berkes.2006}.\par
First, in \thref{lemma1} we show that with $\hat{k}$ as in (\ref{bp_estimator_wilcoxon}),
\[T_n(X_1,\ldots,X_{\hat{k}}) = T_n(Y_1,\ldots,Y_{\hat{k}})+o_P(1)\] and
\[T_n(X_{\hat{k}+1},\ldots,X_{n})=T_n(Y_{\hat{k}+1},\ldots,Y_{n})+o_P(1).\]
Subsequently, in \thref{lemma2} we prove that 
\[
\Big(T_n(Y_1,\ldots,Y_{\hat{k}}),T_n(Y_{\hat{k}+1},\ldots,Y_{n})\Big)\xrightarrow{d}  \sigma (Z^{(1)},Z^{(2)}),
\]
where $Z^{(i)}=\sup_{0\leq t \leq 1}|B^{(i)}(t)|$, $i=1,2$.
Then, the claim (\ref{claim_th}) of \thref{theorem_hypothesis} follows by the continuous mapping theorem. 
\end{proof}\par \bigskip

\subsubsection{Auxiliary results}\label{aux_res}

In this section we state auxiliary results needed to prove \thref{lemma1} and \thref{lemma2} in sections \ref{sec_lemma_1} and \ref{sec_lemma_2}, respectively.

\subsubsection*{Concept of $1$-continuity}

Before we state the auxiliary results, we recall the concept of \textit{$1$-continuity}, which was introduced by \citet{Borovkova.2001}.\par \bigskip

To study the asymptotic behaviour of the Wilcoxon test 
\[W_{1,n}(k)=\sum_{i=1}^{k}\sum_{j=k+1}^{n}(1_{\{X_i\leq X_j\}}-1/2)\]
we need to show that the function $h(x,y)=1_{\{x\leq y\}}$ is $1$-continuous. Then the variables $(h(Y_i,Y_j))$ retain some characteristics of the variables $(Y_i,Y_j)$.

\begin{definition}{(\citet{Borovkova.2001})}\thlabel{1-continuous}\par
We say that the kernel $h\left(x,y\right)$ is $1$-continuous with respect to a distribution of a stationary process $(Y_j)$ if there exists a function $\phi(\epsilon)$, $\epsilon\geq 0$ such that $\phi\left(\epsilon\right)\rightarrow 0$, $\epsilon\rightarrow 0$, and for all $\epsilon > 0$ and $k\geq 1$
\begin{align}
\E \left(\left| h\left(Y_1,Y_k\right)-h\left(Y_1',Y_k\right) \right| 1_{\left\{\left|Y_1-Y_1'\right|\leq \epsilon\right\}}\right) & \leq \phi\left(\epsilon\right),\label{1-cont_11}\\
\E \left(\left| h\left(Y_k,Y_1\right)-h\left(Y_k,Y_1'\right) \right| 1_{\left\{\left|Y_1-Y_1'\right|\leq \epsilon\right\}}\right) & \leq \phi\left(\epsilon\right),\nonumber
\end{align}
and
\begin{align}
\E \left(\left| h\left(Y_1,Y_2'\right)-h\left(Y_1',Y_2'\right) \right| 1_{\left\{\left|Y_1-Y_1'\right|\leq \epsilon\right\}}\right) & \leq \phi\left(\epsilon\right),\label{1-cont_22}\\
\E \left(\left| h\left(Y_2',Y_1\right)-h\left(Y_2',Y_1'\right) \right| 1_{\left\{\left|Y_1-Y_1'\right|\leq \epsilon\right\}}\right) & \leq \phi\left(\epsilon\right),\nonumber
\end{align}
where $Y_2'$ is an independent copy of $Y_1$ and $Y_1'$ is any random variable that has the same distribution as $Y_1$.
\end{definition}

For a univariate function $g(x)$, the $1$-continuity property is defined as follows.

\begin{definition}\par
The function $g\left(x\right)$ is $1$-continuous with respect to a distribution of a stationary process $(Y_j)$ if there exists a function $\phi(\epsilon)$, $\epsilon\geq 0$ such that $\phi\left(\epsilon\right)\rightarrow 0$, $\epsilon\rightarrow 0$, and for all $\epsilon > 0$
\begin{align}\label{1-cont_3}
\E \left(\left| g\left(Y_1\right)-g\left(Y_1'\right) \right| 1_{\left\{\left|Y_1-Y_1'\right|\leq \epsilon\right\}}\right) & \leq \phi\left(\epsilon\right),
\end{align}
where $Y_1'$ is any random variable that has the same distribution as $Y_1$.
\end{definition}
\bigskip

Note that the term $W_{1,n}(k)$ can be written as a second order U-statistic
\begin{equation*}
 U_{a,b}\left(k\right) = \sum_{i=a}^{k}\sum_{j=k+1}^{b}\left(h\left(Y_i,Y_j\right)-\Theta\right), \qquad a\leq k < b,
\end{equation*}
with kernel function $h\left(x,y\right) = 1_{\left\{x\leq y\right\}}$ and constant $\Theta=\E h(Y_1',Y_2')=1/2$, where $Y_1'$ and $Y_2'$ are independent copies of $Y_1$.\par 
By applying Hoeffding's decomposition of U-statistics (\citet{Hoeffding.1948}) to $U_{a,b}(k)$, the kernel function $h$ can be written as the sum
\begin{equation}\label{hoeffding_decomposition_h}
h\left(x,y\right) = \Theta + h_1\left(x\right) + h_2\left(y\right) + g\left(x,y\right),
\end{equation}
where $h_1\left(x\right) = \E h\left(x,Y_2'\right) - \Theta = 1/2 - F\left(x\right),$
\begin{align*}
 \qquad h_2\left(y\right)= \E h\left(Y_1',y\right) - \Theta = F\left(y\right) - 1/2, \qquad
g\left(x,y\right) = h\left(x,y\right) -   h_1\left(x\right) - h_2\left(y\right)-\Theta.
\end{align*}

The following remark states that the bounded functions $h(x,y)=1_{\{x\leq y\}}$, $h_1(x)$, $h_2(x)$ and $g(x,y)$ are $1$-continuous functions.\par

\begin{rremark}\thlabel{remark_1-cont} Let $(Y_j)$ be a stationary process, $Y_1$ has continuous distribution function $F$ with bounded second derivative and the variables $Y_1-Y_k$, $k\geq 1$ satisfy (\ref{Bed_gem_Verteilung}).\par 
\begin{enumerate}[i)]
\item  The function $h(x,y)=1_{\{x\leq y\}}$ is $1$-continuous function (i.e. satisfies (\ref{1-cont_11}) and (\ref{1-cont_22})) with respect to the distribution of $(Y_j)$ with function $\phi(\epsilon)=C\epsilon$, for some $C>0$, see e.g. Corollary 4.1 of \citet{Gerstenberger.2016}.\par 
\item Lemma 2.15 of \citet{Borovkova.2001} yields that if a general function $h(x,y)$ satisfies (\ref{1-cont_11}) and (\ref{1-cont_22}) with some function $\phi(\epsilon)$ then $\E h(x,Y_2')$, where $Y_2'$ is an independent copy of $Y_1$, satisfies the condition in (\ref{1-cont_3}) with the same function $\phi(\epsilon)$. Hence, $h_1(x)=\E h\left(x,Y_2'\right) -1/2$ and $h_2(y)=\E h\left(Y_2',y\right) - 1/2$ are $1$-continuous.\par 
\item  The function $g(x,y)=h(x,y)-h_1(x)-h_2(x)-1/2$ is $1$-continuous (satisfies (\ref{1-cont_11}) and (\ref{1-cont_22})), since $h$ and $h_1$ satisfy (\ref{1-cont_11}), (\ref{1-cont_22}) and (\ref{1-cont_3}) with $\phi(\epsilon)=C\epsilon$, for some $C>0$. In particular,
\begin{align*}
&\E \left(|g(Y_1,Y_k)-g(Y_1',Y_k)|1_{\left\{\left|Y_1-Y_1'\right|\leq \epsilon\right\}} \right)\\
&\leq \E \left(|h(Y_1,Y_k)-h(Y_1',Y_k)|1_{\left\{\left|Y_1-Y_1'\right|\leq \epsilon\right\}} \right) + \E \left(|h_1(Y_1)-h_1(Y_1')|1_{\left\{\left|Y_1-Y_1'\right|\leq \epsilon\right\}} \right)\\
&\leq 2\phi(\epsilon)
\end{align*}
and similarly, $\E \big(|g(Y_k,Y_1)-g(Y_k,Y_1')|1_{\left\{\left|Y_1-Y_1'\right|\leq \epsilon\right\}} \big)\leq 2\phi(\epsilon)$. 
\end{enumerate}
\end{rremark}

\subsubsection*{Auxiliary results}

The following lemma derives the functional central limit theorem for partial sum processes of $(h_1(Y_j))$.

\begin{lemma}\thlabel{remark_h1}
Suppose that the assumptions of \thref{lemma2} hold. Then,
\begin{equation*}
\bigg( \frac{1}{n^{1/2}}\sum_{i=1}^{\left[nt\right]}h_1\left(Y_i\right) \bigg)_{0\leq t \leq 1} \xrightarrow{d} \left( \sigma W\left(t\right) \right)_{0\leq t \leq 1},
\end{equation*}
where $W\left(t\right)$ is a Brownian motion and $\sigma$ is given in (\ref{lrv_sigma2}).
\end{lemma}

\begin{proof}
\citet{Wooldridge.1988} in Corollary 3.2 established a functional central limit theorem for partial sum process $\sum_{i=1}^k \tilde{Y}_i$, $k\geq 1$, for a process $(\tilde{Y}_j)$ which is $L_2$ NED on a strongly mixing process $(\tilde{Z}_j)$. Therefore, \thref{remark_h1} is proved, by showing that  $(h_1(Y_j))$ is $L_2$ NED on a strongly mixing process.\par 
By Proposition 2.11 of \citet{Borovkova.2001}, if $(Y_j)$ is $L_1$ NED on a stationary absolutely regular process $(Z_j)$ with approximation constants $a_k$ and $g(x)$ is $1$-continuous with function $\phi$, then $(g(Y_j))$ is also $L_1$ NED on $(Z_j)$ with approximation constants $a_k'=\phi\left(\sqrt{2a_k}\right)+2\sqrt{2a_k}||g||_{\infty}$. By \thref{remark_1-cont} ii), $h_1(x)=1/2-F(x)$ is $1$-continuous function with $\phi(\epsilon)=C\epsilon$. Thus, the processes $(h_1(Y_j))$ is $L_1$ NED processes with approximation constants $a'_k=C\sqrt{a_k}\geq \phi\left(\sqrt{2a_k}\right)+2\sqrt{2a_k}||h_1||_{\infty}$.\par 
Observe that the variables $\eta_k := h_1(Y_1)-\E(h_1(Y_1)|\mathcal{G}_{-k}^k)$ satisfy the $L_1$ NED condition (\ref{ned_condition}) with $a'_k$. To show $L_2$ NED for $(h_1(Y_j))$ note that by definition of $h_1$, $\E h_1(Y_1)=0$ and $|h_1(Y_1)|\leq C<\infty$. Thus,
\[
\E\eta_k^2 \leq \E\Big(|\eta_k| \cdot (|h_1(Y_1)|+|\E (h_1(Y_1)|\mathcal{G}_{-k}^k)|)\Big) \leq C\E|\eta_k|\leq Ca'_k.
\]
The last inequality holds, because by $L_1$ NED of $(h_1(Y_j))$, $\E|h_1(Y_1)-\E(h_1(Y_1)|\mathcal{G}_{-k}^k)|\leq a_k'$. Therefore, the process $(h_1(Y_j))$ is also $L_2$ NED on $(Z_j)$ with approximation constant $a_k'=Ca_k^{1/2}$. Moreover, absolute regularity of $(Z_j)$ implies the process $(Z_j)$ is also strong mixing. Assumption (\ref{cond_a_beta}) yields $a'_k=O(k^{-1/2})$ and $\beta_k=O(k^{-2})$. Thus, $(h_1(Y_j))$ satisfies the conditions of Corollary 3.2 of \citet{Wooldridge.1988} which proves the lemma.
\end{proof}
\par\bigskip

Next we show that the contribution of $g(x,y)$ of the Hoeffding decomposition (\ref{hoeffding_decomposition_h}) is negligible.

\begin{lemma}\thlabel{g_negligible2}
Suppose that the assumptions of \thref{lemma2} hold. Then,
\begin{equation}\label{g_negligible2_1}
n^{-3/2}\max_{1\leq k \leq n}\max_{1\leq l \leq n} \Big| \sum_{i=1}^{k}\sum_{j=1}^{l}g(Y_i,Y_j)\Big| =o_P(1).
\end{equation}
\end{lemma}

\begin{proof}
We first prove for $1 \leq q \leq p \leq n$, $1 \leq h \leq l \leq n$,
\begin{equation}\label{remark_lemma1_gilt}
\E\Big(\Big|n^{-3/2}\sum_{i=q+1}^{p}\sum_{j=h+1}^{l}g(Y_i,Y_j)\Big|^2\Big)\leq \frac{C}{n^{3}}(p-q)(l-h).
\end{equation}

\textit{Proof of (\ref{remark_lemma1_gilt})} Lemma 1 of \citet{Dehling.2015} showed if $f$ is a $1$-continuous bounded degenerate kernel function and $\phi_f(\epsilon)$ satisfies 
\begin{equation}\label{condition_beta_a_phi}
\sum_{k=1}^{\infty}k(\beta(k)+\sqrt{a_k}+\phi_f(a_k))<\infty,
\end{equation}
then
\begin{equation}\label{lemma1_DFGW}
\E\Big(\sum_{i=1}^{k}\sum_{j=k+1}^{n}f(Y_i,Y_j)\Big)^2\leq Ck(n-k), \qquad 1\leq k \leq n,
\end{equation}
where the constant $C$ depends on the left hand side of (\ref{condition_beta_a_phi}). The proof of Lemma 1 in \citet{Dehling.2015} shows that (\ref{lemma1_DFGW}) can be extended to (\ref{remark_lemma1_gilt}). Hence, to complete the proof, we need to verify that $g(x,y)$ satisfies the assumptions of Lemma 1 of \citet{Dehling.2015}.\par 
By the Hoeffding decomposition (\ref{hoeffding_decomposition_h}), $g(x,y)=h(x,y)+F(x)-F(y)-1/2$. Note that $\E F(Y_1)=1/2$, thus $\E g(x,Y_1)=\E g(Y_1,y)=0$, i.e. $g(x,y)$ is a degenerate kernel. Furthermore, $g(x,y)$ is bounded, since $h(x,y)=1_{\{x\leq y\}}$ and $F(x)$ are bounded. By \thref{remark_1-cont} iii) $g(x,y)$ is $1$-continuous with $\phi(\epsilon)=C\epsilon$, the latter satisfies (\ref{condition_beta_a_phi}) because of condition (\ref{cond_a_beta}). This completes the proof of (\ref{remark_lemma1_gilt}).
\par\bigskip

\textit{Proof of (\ref{g_negligible2_1})} To prove the lemma, we use Theorem 10.2 of \citet{Billingsley.1999}, which states that if the increments of partial sums $S_i = \sum_{j=1}^{i}\zeta_i$ of random variables $\zeta_i$, $i=1,2,\ldots$ are bounded in probability, in particular if there exist $\alpha > 1$, $\beta>0$ and non-negative numbers $u_{n,1},\ldots,u_{n,n}$ such that
\begin{equation*}%\label{condition_prop_billingsley}
\p\big(\left|S_j-S_i\right|\geq \epsilon\big)\leq \frac{1}{\epsilon^{\beta}}\bigg(\sum_{l=i+1}^{j}u_{n,l}\bigg)^{\alpha},
\end{equation*}
for $\epsilon>0$, $0\leq i \leq j \leq n$, then for all $\epsilon>0$, $n\geq 2$,
\begin{equation*}%\label{result_prop_billingsley}
\p\left(\max_{1\leq  k \leq n}\left|S_k\right| \geq \epsilon\right) \leq \frac{K}{\epsilon^{\beta}}\bigg(\sum_{l=1}^{n}u_{n,l}\bigg)^{\alpha},
\end{equation*}
where $K>0$ depends only on $\alpha$ and $\beta$.\par 

Denote 
\begin{equation*}%\label{g_negligible_hilsfunktion}
G_n(l) = n^{-3/2}\max_{1\leq k\leq n}\Big|\sum_{i=1}^{k}\sum_{j=1}^{l}g(Y_i,Y_j)\Big|,
\end{equation*}
with $G_n(0)=0$ and define random variables $\zeta_i = G_n(i)-G_n(i-1)$, where $\zeta_0=0$. Note that $S_i = \sum_{j=1}^{i}\zeta_i = G_n(i)$ and by using the reverse triangle inequality, for $1\leq h \leq l \leq n$,
\begin{align*}
\p\left(\left|S_l-S_h\right|\geq \epsilon\right) &\leq \p\Big(n^{-3/2}\max_{1\leq k\leq n}\Big|\sum_{i=1}^{k}\sum_{j=1}^{l}g(Y_i,Y_j)-\sum_{i=1}^{k}\sum_{j=1}^{h}g(Y_i,Y_j)\Big|\geq \epsilon\Big)\\
&=\p\Big(n^{-3/2}\max_{1\leq k\leq n}\Big|\sum_{i=1}^{k}\sum_{j=h+1}^{l}g(Y_i,Y_j)\Big|\geq \epsilon\Big).
\end{align*}
Let us now define
\[
\tilde{S}_k = \sum_{i=1}^{k}\Big(n^{-3/2}\sum_{j=h+1}^{l}g(Y_i,Y_j)\Big),
\]
and note that $\tilde{S}_k$ depends on $h$ and $l$. Furthermore, note that for $1 \leq q \leq p \leq n$,
\[
\big|\tilde{S}_p-\tilde{S}_q\big| =  n^{-3/2}\Big|\sum_{i=q+1}^{p}\sum_{j=h+1}^{l}g(Y_i,Y_j)\Big|.
\]
By Markov inequality and (\ref{remark_lemma1_gilt}),
\[
\p\Big(\big|\tilde{S}_p-\tilde{S}_q\big|\geq \epsilon\Big) \leq \frac{1}{\epsilon^2}\E\Big(\big|\tilde{S}_p-\tilde{S}_q\big|^2\Big)\leq \frac{1}{\epsilon^2}\frac{C}{n^{3}}(p-q)(l-h)\leq \frac{1}{\epsilon^2}\Big(\sum_{t=q+1}^{p}u_{n,t}\Big)^{4/3},
\]
where $u_{n,t}=\frac{C^{3/4}}{n^{9/4}}(l-h)$. Hence, $\tilde{S}_i$ satisfies assumption of Theorem 10.2 of \citet{Billingsley.1999} with $\beta=2$, $\alpha=4/3$. Thus,  for any fixed $\epsilon>0$,
\begin{equation*}
\p\Big(\max_{1\leq k\leq n}\big|\tilde{S}_k\big|\geq \epsilon\Big)
 \leq \frac{K}{\epsilon^2}\bigg(\sum_{t=1}^{n}\frac{C^{3/4}}{n^{9/4}}(l-h)\bigg)^{4/3}\leq \frac{1}{\epsilon^2}\bigg((l-h)\frac{C^{3/4}}{n^{5/4}}\bigg)^{4/3}
\end{equation*}
and moreover
\[
\p\left(\left|S_l-S_h\right|\geq \epsilon\right)\leq \p\Big(\max_{1\leq k\leq n}\big|\tilde{S}_k\big|\geq \epsilon\Big) \leq \frac{1}{\epsilon^2}\bigg(\sum_{t=h+1}^{l}u_{n,t}\bigg)^{4/3},
\]
where $u_{n,t}=\frac{C^{3/4}}{n^{5/4}}$. Therefore, $S_i$ satisfies assumption of Theorem 10.2 of \citet{Billingsley.1999} with $\beta=2$, $\alpha=4/3$. Finally, for any fixed $\epsilon>0$, as $n\rightarrow\infty,$
\begin{multline*}
\p\Big(n^{-3/2}\max_{1\leq l \leq n} \max_{1\leq k \leq n}\Big| \sum_{i=1}^{k}\sum_{j=1}^{l}g(Y_i,Y_j)\Big|\geq \epsilon\Big)\\
=\p\Big(\max_{1\leq l \leq n}\big|S_l\big|\geq \epsilon\Big)\leq \frac{K}{\epsilon^2}\bigg(\sum_{t=1}^{n}\frac{C^{3/4}}{n^{5/4}}\bigg)^{4/3} \leq \frac{K}{\epsilon^2}\frac{1}{n^{1/3}} \rightarrow 0,
\end{multline*}
which proves the lemma.
\end{proof}
\par\bigskip

\par\bigskip
In the following we state auxiliary results to deal with the terms
\begin{equation*}
\tilde{U}_{1,\hat{k}}(k^*):=\sum_{i=1}^{k^*}\sum_{j=k^*+1}^{\hat{k}} 1_{\{Y_j < Y_i \leq Y_j+\Delta_n\}}, \qquad \hat{k}\geq k^*=[n\theta],
\end{equation*}
and
\begin{equation*}
\tilde{U}_{\hat{k}+1,n}(k^*):=\sum_{i=\hat{k}+1}^{k^*}\sum_{j=k^*+1}^{n} 1_{\{Y_j < Y_i \leq Y_j+\Delta_n\}}, \qquad \hat{k}<k^*
\end{equation*}
appearing in the proof of \thref{lemma1}.\par
Note that the terms $\tilde{U}_{1,\hat{k}}(k^*)$ and $\tilde{U}_{\hat{k}+1,n}(k^*)$ can be written as a second order U-statistic
\begin{equation*}
 \tilde{U}_{a,b}\left(k\right) = \sum_{i=a}^{k}\sum_{j=k+1}^{b}h_n\left(Y_i,Y_j\right), \qquad a\leq k < b,
\end{equation*}
with kernel function $h_n\left(x,y\right) = 1_{\{y < x \leq y+\Delta_n\}}$.\par 
Applying Hoeffding's decomposition of U-statistics to $\tilde{U}_{a,b}(k)$, decomposes the kernel function $h_n$ into the sum
\begin{equation}\label{hoeffding_decomposition_hn}
h_n\left(x,y\right) = \Theta_{\Delta_n} + h_{1,n}\left(x\right) + h_{2,n}\left(y\right) + g_n\left(x,y\right),
\end{equation}
with $\Theta_{\Delta_n} = \E \big(1_{\{Y_2' < Y_1' \leq Y_2'+\Delta_n\}}\big)$,
\begin{align*}
h_{1,n}\left(x\right) &= \E h_n\left(x,Y_2'\right) - \Theta_{\Delta_n} = F\left(x\right) - F\left(x-\Delta_n\right) - \Theta_{\Delta_n}, \\
h_{2,n}\left(y\right)&= \E h_n\left(Y_1',y\right) - \Theta_{\Delta_n}  = F\left(y+\Delta_n\right) - F\left(y\right) - \Theta_{\Delta_n},\\
g_n\left(x,y\right) &= h_n\left(x,y\right) -   h_{1,n}\left(x\right) - h_{2,n}\left(y\right)-\Theta_{\Delta_n},
\end{align*}
where $Y_1'$ and $Y_2'$ are independent copies of $Y_1$.\par \bigskip

\begin{lemma}\thlabel{Lemma_abschaetzung_bp_kdach}
Suppose that the assumptions of \thref{lemma1} hold. Then,
\begin{equation}\label{Lemma_abschaetzung_bp_kdach1}
n^{-3/2}\Big| \tilde{U}_{1,\hat{k}}(k^*) -k^*(\hat{k}-k^*)\Theta_{\Delta_n} \Big| = o_P\left(1\right)
\end{equation}
and
\begin{equation}\label{Lemma_abschaetzung_bp_kdach2}
n^{-3/2}\Big| \tilde{U}_{\hat{k}+1,n}(k^*) -(k^*-\hat{k})(n-k^*)\Theta_{\Delta_n} \Big| = o_P\left(1\right),
\end{equation}
where $Y_1'$ and $Y_2'$ are independent copies of $Y_1$ .
\end{lemma}

\begin{proof}
Let us start with the proof of (\ref{Lemma_abschaetzung_bp_kdach1}). The Hoeffding decomposition (\ref{hoeffding_decomposition_hn}) yields
\begin{align*}
\tilde{U}_{1,\hat{k}}(k^*) -k^*(\hat{k}-k^*)\Theta_{\Delta_n} = \sum_{i=1}^{k^*}\sum_{j=k^*+1}^{\hat{k}} (h_{1,n}\left(Y_i\right) + h_{2,n}\left(Y_j\right) + g_n\left(Y_i,Y_j\right))\\
=(\hat{k}-k^*)\sum_{i=1}^{k^*}h_{1,n}\left(Y_i\right) + k^*\sum_{j=k^*+1}^{\hat{k}} h_{2,n}\left(Y_j\right) + \sum_{i=1}^{k^*}\sum_{j=k^*+1}^{\hat{k}}g_n\left(Y_i,Y_j\right).
\end{align*}
Therefore,
\begin{multline*}
n^{-3/2}\Big| \tilde{U}_{1,\hat{k}}(k^*) -k^*(\hat{k}-k^*)\Theta_{\Delta_n} \Big|\\ \leq n^{-3/2}\Big|(\hat{k}-k^*)\sum_{i=1}^{k^*}h_{1,n}\left(Y_i\right) + k^*\sum_{j=k^*+1}^{\hat{k}} h_{2,n}\left(Y_j\right)\Big| + n^{-3/2}\Big| \sum_{i=1}^{k^*}\sum_{j=k^*+1}^{\hat{k}}g_n\left(Y_i,Y_j\right)\Big|.
\end{multline*}
Note that the indicator function $h_n(x,y)=1_{\{y < x \leq y+\Delta_n\}}$ is bounded.\\
The distribution function $F$ of $Y_1$ has bounded second derivative. Hence, as $n\rightarrow\infty,$
\begin{align}
\Theta_{\Delta_n} &= \E 1_{\{Y_2' < Y_1' \leq Y_2'+\Delta_n\}} = \p\left(Y_2' < Y_1' \leq Y_2'+\Delta_n\right)\nonumber\\ 
&= \int_{\mathbb{R}}\left(F\left(y+\Delta_n\right)-F(y)\right)dF(y) = \Delta_n \bigg(\int_{\mathbb{R}}f^2\left(y\right)dy + o(1)\bigg)\sim C\Delta_n. \label{remark_theta_delta_n}
\end{align}
Thus,
\begin{align}
|h_{1,n}(x)| &\leq |F(x)-F(x-\Delta_n)-\Theta_{\Delta_n}|\leq C\Delta_n + \Theta_{\Delta_n} \leq C\Delta_n,\label{abschaetzung_h1n_kleiner_delta}\\
|h_{2,n}(x)| &\leq |F(x+\Delta_n)-F(x)-\Theta_{\Delta_n}|\leq C\Delta_n + \Theta_{\Delta_n} \leq C\Delta_n,\nonumber
\end{align}
where $C>0$ is a constant. Hence, $g_n\left(x,y\right) = h_n\left(x,y\right) - h_{1,n}\left(x\right)- h_{2,n}\left(y\right)-\Theta_{\Delta_n}$ is bounded. Since $\E h_{1,n}(Y_1)=0$ and $\E h_{2,n}(Y_1)=0$, $g_n(x,y)$ is a degenerate kernel, i.e.  $\E g_n(x,Y_1)=\E g_n(Y_1,y)=0$. $h_n(x,y)$ satisfies (\ref{1-cont_11}) and (\ref{1-cont_22}) with $\phi_{h_n}(\epsilon)=C\epsilon$, see e.g. Corollary 4.1 of \citet{Gerstenberger.2016}, where constant $C$ does not depend on $n$. Then, with similar argument as in \thref{remark_1-cont}, $h_{1,n}$ and $h_{2,n}$ are $1$-continuous and therefore, $g_n(x,y)$ is $1$-continuous with function $\phi_{g_n}(\epsilon)=C\epsilon$ satisfying (\ref{condition_beta_a_phi}). Hence, $g_n(x,y)$ satisfies the conditions on $g(x,y)$ in \thref{g_negligible2}, which yields
\begin{equation*}
n^{-3/2}\Big| \sum_{i=1}^{k^*}\sum_{j=k^*+1}^{\hat{k}}g_n\left(Y_i,Y_j\right)\Big|\leq 2\max_{1\leq k \leq n}\max_{1\leq k \leq n}n^{-3/2}\Big| \sum_{i=1}^{k}\sum_{j=1}^{l}g_n\left(Y_i,Y_j\right)\Big|=o_P(1).
\end{equation*}

Thus, it remains to show $n^{-3/2}\big|(\hat{k}-k^*)\sum_{i=1}^{k^*}h_{1,n}\left(Y_i\right) + k^*\sum_{j=k^*+1}^{\hat{k}} h_{2,n}\left(Y_j\right)\big|=o_P(1)$.\par
By (\ref{abschaetzung_h1n_kleiner_delta}), we receive the following inequality
\begin{multline*}
n^{-3/2}\Big|(\hat{k}-k^*)\sum_{i=1}^{k^*}h_{1,n}\left(Y_i\right) + k^*\sum_{j=k^*+1}^{\hat{k}} h_{2,n}\left(Y_j\right)\Big|\\
\leq n^{-3/2}C(\hat{k}-k^*)k^*\Delta_n = C\frac{k^*}{n}\frac{\Delta_n^2|\hat{k}-k^*|}{n^{1/2}\Delta_n}=o_P(1),
\end{multline*}
where we used the consistency of $\hat{k}$ in (\ref{rate_consistency_bp_estimator}), $\Delta_n^2|\hat{k}-k^*|=O_P(1),$ and Assumption \ref{H2}, $k^*/n\sim\theta$ and $n\Delta_n^2\rightarrow\infty$ as $n\rightarrow\infty$. This completes the proof of (\ref{Lemma_abschaetzung_bp_kdach1}).\par 
The proof of (\ref{Lemma_abschaetzung_bp_kdach2}) follows using similar argument.
\end{proof}

\subsubsection{Proof of \thref{lemma1}}\label{sec_lemma_1}

Before proceeding to \thref{lemma1}, similarly to the notation $W_{m,n}(k)$ in (\ref{Wilcoxon-TS}), we define
\begin{equation}\label{define_U_Y}
U_{m,n}(k)= \sum_{i=m}^{k}\sum_{j=k+1}^{n}(1_{\left\{Y_i\leq Y_j\right\}}-1/2), \qquad m\leq k \leq n.
\end{equation}
Note that $W_{m,n}(k)$ depends on $(X_m, \ldots, X_n)$, where $U_{m,n}(k)$ depends on $(Y_m, \ldots, Y_n)$.

\begin{nono-lemma}
Let $X_1,\ldots,X_n$ follow the model in (\ref{model_srd}), and Assumptions \ref{H1} and \ref{H2} be satisfied. Let $\hat{k}$ be defined as in (\ref{bp_estimator_wilcoxon}). Then,
\begin{align}
n^{-3/2}\max_{1\leq k \leq \hat{k}} \big|W_{1,\hat{k}}(k) \big| 
&= n^{-3/2}\max_{1\leq k \leq \hat{k}} \big| U_{1,\hat{k}}(k) \big| + o_P\left(1\right)\label{lemma1_1}\\
n^{-3/2}\max_{\hat{k}< k \leq n} \big| W_{\hat{k}+1,n}(k) \big|
&= n^{-3/2}\max_{\hat{k}< k \leq n} \big| U_{\hat{k}+1,n}(k) \big| + o_P\left(1\right).\label{lemma1_2}
\end{align}
\end{nono-lemma}

\begin{proof}

We have to distinguish between two cases, $\hat{k}\leq k^*$ and $\hat{k}>k^*$, where $k^*=[n\theta]$.\par 

If $\hat{k}\leq k^*$, then by (\ref{model_srd}), $X_i=Y_i$, $i=1,\ldots,\hat{k}$, and hence, $W_{1,\hat{k}}(k) = U_{1,\hat{k}}(k)$, $k=1,\ldots,\hat{k}$. In turn, $X_i=Y_i$ for $i=\hat{k}+1,\ldots,k^*$, and $X_i=Y_i+\Delta_n$ for $i=k^*+1,\ldots,n$. Since $1_{\left\{Y_i+\Delta_n\leq Y_j+\Delta_n\right\}}=1_{\left\{Y_i\leq Y_j\right\}}$,  $W_{\hat{k}+1,n}(k)$ can be decomposed into two terms,
\begin{align*}
W_{\hat{k}+1,n}(k)= \begin{cases}
U_{\hat{k}+1,n}(k) + \sum_{i=\hat{k}+1}^{k}\sum_{j=k^*+1}^{n} 1_{\{Y_j < Y_i \leq Y_j+\Delta_n\}},\quad \hat{k}< k \leq k^*\\
U_{\hat{k}+1,n}(k) + \sum_{i=\hat{k}+1}^{k^*}\sum_{j=k+1}^{n} 1_{\{Y_j < Y_i \leq Y_j+\Delta_n\}},\quad k^* < k \leq n.
\end{cases}
\end{align*}
If $\hat{k}>k^*$, similar argument yields, $W_{\hat{k}+1,n}(k)=U_{\hat{k}+1,n}(k)$, for $k=\hat{k}+1,\ldots,n$ and 
\begin{align}\label{lemma1_11}
W_{1,\hat{k}}(k) = \begin{cases}
U_{1,\hat{k}}(k) + \sum_{i=1}^{k}\sum_{j=k^*+1}^{\hat{k}} 1_{\{Y_j < Y_i \leq Y_j+\Delta_n\}},\quad 1\leq k \leq k^*\\
U_{1,\hat{k}}(k) + \sum_{i=1}^{k^*}\sum_{j=k+1}^{\hat{k}} 1_{\{Y_j < Y_i \leq Y_j+\Delta_n\}},\quad k^* < k \leq \hat{k}.
\end{cases}
\end{align}
\textit{Proof of (\ref{lemma1_1})}. For $\hat{k}\leq k^*$, equation (\ref{lemma1_1}) holds trivially, since $W_{1,\hat{k}}(k) = U_{1,\hat{k}}(k)$, $k=1,\ldots,\hat{k}$.\par 
For $\hat{k} > k^*$, equation (\ref{lemma1_11}) yields,
\begin{equation*}
\Big|W_{1,\hat{k}}(k) - U_{1,\hat{k}}(k)\Big| \leq \sum_{i=1}^{k^*}\sum_{j=k^*+1}^{\hat{k}} 1_{\{Y_j < Y_i \leq Y_j+\Delta_n\}} =: I_{1,\hat{k}}(k^*),
\end{equation*}
for all $1\leq k \leq \hat{k}$.
Hence, using the reverse triangle inequality,
\begin{align*}
\Big| n^{-3/2}\max_{1\leq k \leq \hat{k}}\big| W_{1,\hat{k}}(k)\big| - n^{-3/2}\max_{1\leq k \leq \hat{k}}\big| U_{1,\hat{k}}(k)\big| \Big|\leq n^{-3/2}I_{1,\hat{k}}(k^*).
\end{align*}
Thus, property (\ref{lemma1_1}) holds if $n^{-3/2}I_{1,\hat{k}}(k^*)=o_P(1)$.\par 
By \thref{Lemma_abschaetzung_bp_kdach}, $n^{-3/2}I_{1,\hat{k}}(k^*) = n^{-3/2}k^*(\hat{k}-k^*)\Theta_{\Delta_n}+o_P(1),$
where $\Theta_{\Delta_n}= \E \big(1_{\{Y_2' < Y_1' \leq Y_2'+\Delta_n\}}\big)$ and $Y_1'$ and $Y_2'$ are independent copies of $Y_1$. The distribution function $F$ of $Y_1$ has bounded second derivative. Hence, as $n\rightarrow\infty$, by (\ref{remark_theta_delta_n}),
\begin{align*}
\Theta_{\Delta_n} = \Delta_n \bigg(\int_{\mathbb{R}}f^2\left(y\right)dy + o(1)\bigg).
\end{align*}
Furthermore, by (\ref{rate_consistency_bp_estimator}), $\Delta_n^2|\hat{k}-k^*|=O_P(1)$ and by Assumption \ref{H2}, $k^*/n\sim\theta$ and $n\Delta_n^2\rightarrow\infty$, as $n\rightarrow\infty$. This yields
\begin{equation*}
n^{-3/2}k^*|\hat{k}-k^*|\Theta_{\Delta_n}\leq  C\frac{\Delta_n^2\big|\hat{k}-k^*\big|}{n^{1/2}\Delta_n} =o_P(1).
\end{equation*}
This completes the proof of (\ref{lemma1_1}). The proof of (\ref{lemma1_2}) follows using similar argument.
\end{proof}

\subsubsection{Proof of \thref{lemma2}}\label{sec_lemma_2}
We will now state the proof of \thref{lemma2}.
\begin{proof} To prove \thref{lemma2} we will use the idea of the proof of  Theorem 3 of \citet{Dehling.2015}. \par 
Recall that $T(Y_1,\ldots,Y_{\hat{k}})=\hat{k}^{-3/2}\max_{1\leq k \leq \hat{k}}|U_{1,\hat{k}}(k)|$ and similarly $T(Y_{\hat{k}+1},\ldots,Y_{n})=(n-\hat{k})^{-3/2}\max_{\hat{k}<k\leq n}|U_{\hat{k}+1,n}(k)|$.
Note that the terms $U_{1,\hat{k}}(k)$ and $U_{\hat{k}+1,n}(k)$ defined in (\ref{define_U_Y}) can be written as a second order U-statistic
\begin{equation*}
 U_{a,b}\left(k\right) = \sum_{i=a}^{k}\sum_{j=k+1}^{b}\left(h\left(Y_i,Y_j\right)-\Theta\right), \qquad a\leq k < b,
\end{equation*}
with kernel function $h\left(x,y\right) = 1_{\left\{x\leq y\right\}}$ and constant $\Theta=\E h(Y_1',Y_2')=1/2$, where $Y_1'$ and $Y_2'$ are independent copies of $Y_1$. Furthermore, we can apply the Hoeffding's decomposition given in (\ref{hoeffding_decomposition_h}).

Therefore,
\[
U_{a,b}(k) = \sum_{i=a}^{k}\sum_{j=k+1}^{b}\left(h_1\left(Y_i\right)+h_2\left(Y_j\right)+g\left(Y_i,Y_j\right)\right)=: s_{a,b}(k) + v_{a,b}(k),
\]
where
\[
s_{a,b}(k)=(b-k)\sum_{i=a}^{k}h_1\left(Y_i\right) + (k-a+1)\sum_{j=k+1}^{b}h_2\left(Y_j\right), \qquad v_{a,b}(k)=\sum_{i=a}^{k}\sum_{j=k+1}^{b}g\left(Y_i,Y_j\right).
\]

Note that 
\[
v_{a,b}(k)= \sum_{i=1}^{k}\sum_{j=1}^{b}g\left(Y_i,Y_j\right)-\sum_{i=1}^{k}\sum_{j=1}^{k}g\left(Y_i,Y_j\right)-\sum_{i=1}^{a-1}\sum_{j=1}^{b}g\left(Y_i,Y_j\right)+\sum_{i=1}^{a-1}\sum_{j=1}^{k}g\left(Y_i,Y_j\right).
\]

Thus, \thref{g_negligible2} yields
\[
n^{-3/2}\max_{a\leq k \leq b}\big|v_{a,b}(k)\big|\leq 4n^{-3/2}\max_{1\leq k \leq n}\max_{1\leq l \leq n}\Big|  \sum_{i=1}^{k}\sum_{j=1}^{l}g\left(Y_i,Y_j\right)\Big|  =o_P(1).
\]
Furthermore, by the triangle inequality,
\[
\max_{a\leq k \leq b}\big|U_{a,b}(k)\big|= \max_{a\leq k \leq b}\big|s_{a,b}(k)\big| +\max_{a\leq k \leq b}\big|v_{a,b}(k)\big| = \max_{a\leq k \leq b}\big|s_{a,b}(k)\big| + o_P(n^{3/2}).
\]
Consistency of $\hat{k}$ in (\ref{rate_consistency_bp_estimator}), $\Delta_n^2|\hat{k}-k^*|=O_P(1)$, and Assumption \ref{H2}, $n\Delta_n^2\rightarrow\infty$, as $n\rightarrow\infty$, yield
\begin{equation}\label{consistency2}
\Big|\frac{\hat{k}}{n}-\theta\Big| = o_P(1).
\end{equation}
It remains to show that
\begin{align*}
\hat{k}^{-3/2}\max_{1\leq k\leq \hat{k}}\big|s_{1,\hat{k}}(k)\big| &\xrightarrow{d} \sigma\sup_{0\leq t \leq 1}\big|B^{(1)}\left(t\right)\big|, \\
(n-\hat{k})^{-3/2}\max_{\hat{k} < k\leq n}\big|s_{\hat{k}+1,n}(k)\big| &\xrightarrow{d} \sigma\sup_{0\leq t \leq 1}\big|B^{(2)}\left(t\right)\big|,
\end{align*}
where $B^{(1)}$ and $B^{(2)}$ are independent Brownian bridges. By Slutsky's Lemma this implies (\ref{eq_lemma2}). Note that $h_1(x)=-h_2(x)$. Hence,
\begin{align*}
s_{1,\hat{k}}(k) &= (\hat{k}-k)\sum_{i=1}^{k}h_1(Y_i) + k\sum_{j=k+1}^{\hat{k}}h_2(Y_j)\\
&= \hat{k}n^{1/2}\Big\{ \frac{1}{n^{1/2}}\sum_{i=1}^{ k}h_1\left(Y_i\right)-\frac{ k}{\hat{k}}\frac{1}{n^{1/2}}\sum_{i=1}^{\hat{k}}h_1\left(Y_i\right) \Big\} =: \hat{k}n^{1/2}\Gamma_k^{(1)}
\end{align*}
and
\begin{align*}
&s_{\hat{k}+1,n}(k) =(n-k)\sum_{i=\hat{k}+1}^{k}h_1(Y_i) + (k-\hat{k})\sum_{j=k+1}^{n}h_1(Y_j)\\
&= (n-\hat{k})n^{1/2}\Big\{\frac{1}{n^{1/2}}\sum_{i=\hat{k}+1}^{k}h_1\left(Y_i\right)-\frac{k-\hat{k}}{n-\hat{k}}\frac{1}{n^{1/2}}\sum_{i=\hat{k}+1}^{n}h_1\left(Y_i\right)\Big\}\\
&= (n-\hat{k})n^{1/2}\Big\{\frac{1}{n^{1/2}}\Big(\sum_{i=1}^{k}h_1\left(Y_i\right)-\sum_{i=1}^{\hat{k}}h_1\left(Y_i\right)\Big)-\frac{k-\hat{k}}{n-\hat{k}}\frac{1}{n^{1/2}}\Big(\sum_{i=1}^{n}h_1\left(Y_i\right)-\sum_{i=1}^{\hat{k}}h_1\left(Y_i\right)\Big)\Big\}\\
& =:(n-\hat{k})n^{1/2} \Gamma_k^{(2)}.
\end{align*}
\par

\thref{remark_h1} implies weak convergence on $D[0,1]$ of the partial sum process,
\begin{equation*}
\bigg( \frac{1}{n^{1/2}}\sum_{i=1}^{\left[nt\right]}h_1\left(Y_i\right) \bigg)_{0\leq t \leq 1} \xrightarrow{d} \left( \sigma W\left(t\right) \right)_{0\leq t \leq 1},
\end{equation*}
where $W\left(t\right)$ is a Brownian motion and $\sigma$ as in (\ref{lrv_sigma2}). By the Skorokhod-Wichura-Dudley representation (see e.g.,  \citet{Shorack.2009}, Theorem 4 on page 47) there exists a series of Brownian motions $W_n\left(t\right)$, $t\in\left[0,1\right]$, such that
\begin{equation*}
\sup_{0\leq t\leq 1}\Big| n^{-1/2}\sum_{i=1}^{\left[nt\right]}h_1\left(Y_i\right)-\sigma W_n\left(t\right)\Big| = o_P\left(1\right).
\end{equation*}
Set 
\begin{align*}
\Gamma_{W,k}^{(1)} = W_n\Big(\frac{k}{n}\Big)-\frac{k}{\hat{k}}W_n\Big(\frac{\hat{k}}{n}\Big), \quad
\Gamma_{W,k}^{(2)}= \Big(W_n\Big(\frac{k}{n}\Big)- W_n\Big(\frac{\hat{k}}{n}\Big)\Big)-\frac{k-\hat{k}}{n-\hat{k}}\Big(W_n(1)-W_n\Big(\frac{\hat{k}}{n}\Big)\Big),
\end{align*}
and note that $\Gamma_{W,k}^{(1)}$ and $\Gamma_{W,k}^{(2)}$ are independent, since the increments of Brownian motions are independent.\\
Thus,
\begin{equation*}
\max_{1\leq k \leq \hat{k}}\big| \Gamma_{k}^{(1)}- \sigma\Gamma_{W,k}^{(1)}\big| = o_P(1), \qquad \max_{\hat{k}< k \leq n}\big| \Gamma_{k}^{(2)}- \sigma\Gamma_{W,k}^{(2)}\big| = o_P(1).
\end{equation*}
By (\ref{consistency2}) and by the a.s. equicontinuity of the Brownian motion process $\{W_n\}$ and using the continuous mapping theorem, $\big| W_n\big(\hat{k}/n\big)-W_n\left(\theta\right)\big| = o_P\left(1\right)$. Hence,
\begin{equation*}
\max_{1\leq k\leq \hat{k}}\big| \Gamma_{W,k}^{(1)}\big| = \sup_{0\leq t\leq \theta}\Big| W_n\left(t\right)-\frac{t}{\theta}W_n\left(\theta\right)\Big| + o_P\left(1\right)
\end{equation*}
and
\begin{align*}
\max_{\hat{k}< k\leq n}\big| \Gamma_{W,k}^{(2)}\big| &= \sup_{\theta < t\leq 1}\Big| \Big(W_n\left(t\right)-W_n\left(\theta\right)\Big)-\frac{t-\theta}{1-\theta}\Big(W_n\left(1\right)-W_n\left(\theta\right)\Big)\Big| + o_P\left(1\right)\\
& \overset{d}{=}\sup_{\theta < t\leq 1}\Big| W_n\left(t-\theta\right)-\frac{t-\theta}{1-\theta}W_n\left(1-\theta\right)\Big|,
\end{align*}
since Brownian motions have stationary increments and $W_n(0)=0$.
Finally,
\begin{multline*}
(\hat{k}/n)^{-1/2}\max_{1\leq k \leq \hat{k}}\big| \Gamma_{k}^{(1)}\big|
 = \frac{\sigma}{\theta^{1/2}}\sup_{0\leq t\leq \theta}\Big| W_n\left(t\right)-\frac{t}{\theta}W_n\left(\theta\right)\Big|+o_P\left(1\right)\overset{d}{=}\sigma\sup_{0\leq t\leq 1} \big| B^{(1)}\left(t\right)\big|,
\end{multline*}
since Brownian motions are scale invariant, i.e. $\theta^{-1/2}W_n(t)\overset{d}{=}W_n(t/\theta)$, and
\begin{multline*}
((n-\hat{k})/n)^{-1/2}\max_{\hat{k}< k \leq n }\big| \Gamma_{k}^{(2)}\big|\overset{d}{=} \frac{\sigma}{\left(1-\theta\right)^{1/2}}\sup_{\theta<t\leq 1}\Big| W_n\left(t-\theta\right)-\frac{t-\theta}{1-\theta}W_n\left(1-\theta\right)\Big|\\
\overset{d}{=} \frac{\sigma}{\left(1-\theta\right)^{1/2}}\sup_{0<t\leq 1-\theta}\Big| W_n\left(t\right)-\frac{t}{1-\theta}W_n\left(1-\theta\right)\Big|
\overset{d}{=}\sigma\sup_{0\leq t\leq 1} \big| B^{(2)}\left(t\right)\big|.
\end{multline*}
The increments of Brownian motions are independent, thus $B^{(1)}$ and $B^{(2)}$ are independent. This proves the lemma.
\end{proof}

\subsection{Proof of \thref{theorem_alternative}}\label{Proof_alternative}

Under the alternative we consider observations $X_1,\ldots,X_n$ with $X_i=G(\xi_i)+\mu$, $i=1,\ldots,n$. Note that the indicator function $1_{\{x\leq y\}}$ is invariant under strictly increasing functions, i.e. $1_{\{G(\xi_i)\leq G(\xi_j)\}}=1_{\{\xi_i\leq \xi_j\}}$, if $G$ is strictly increasing. For $G$ being a strictly decreasing function, observe that $1_{\{G(\xi_i)\leq G(\xi_j)\}} = 1-1_{\{\xi_i\leq \xi_j\}}$. Therefore, for $G$ being strictly monotone,
\[
\Big|\sum_{i=1}^{k}\sum_{j=k+1}^{n}(1_{\{X_i\leq X_j\}}-1/2)\Big| = \Big|\sum_{i=1}^{k}\sum_{j=k+1}^{n}(1_{\{\xi_i\leq \xi_j\}}-1/2)\Big|.
\]
Thus, to prove \thref{theorem_alternative} it is sufficient to consider $T_{n,1}$ and $T_{n,2}$ in (\ref{T_n1}), (\ref{T_n2}) applied to the stationary Gaussian process $(\xi_j)$, i.e. $T_{n,1}(\xi_1,\ldots,\xi_{\hat{k}})$ and $T_{n,2}(\xi_{\hat{k}+1},\ldots,\xi_n)$,  instead of $T_{n,1}(X_1,\ldots,X_{\hat{k}})$ and $T_{n,2}(X_{\hat{k}+1},\ldots,X_n)$.\par \bigskip

Before we prove that the test $M_n$ tends to infinity in probability under the alternative, we will consider the limit distribution of $T_{n,1}(\xi_1,\ldots,\xi_{\hat{k}})$ and $T_{n,2}(\xi_{\hat{k}+1},\ldots,\xi_n)$ in \thref{lemma_ld_alternative}, using a different normalization $n^{d+3/2}c_d$, where $c_d^2 = \frac{c_0}{d(2d+1)}$, $c_0>0$. Note that in the following we always assume $d\in(0,1/2)$. By $(W_H(t))_{0\leq t \leq 1}$ we denote a fractional Brownian motion process with Hurst parameter $H=d+1/2$, that is a mean zero Gaussian process with auto-covariances $\Cov(W_H(t),W_H(s))= (t^{2H}+s^{2H}-|t-s|^{2H})/2$. 

\begin{lemma}\thlabel{lemma_konv_alternative}
Assume that the assumptions of \thref{theorem_alternative} hold. Then, for $0\leq s\leq t\leq 1$,
\begin{equation*}
\frac{1}{n^{d+3/2}c_d}\sum_{i=1}^{[ns]}\sum_{j=[nt]+1}^n (1_{\{\xi_i\leq \xi_j\}}-1/2) \xrightarrow{d} \frac{1}{2\sqrt{\pi}}\Big(s(W_H(1)-W_H(t))-(1-t)W_H(s)\Big),
\end{equation*}
where $W_H$, $H=d+1/2$ is a standard fractional Brownian motion, $c_d^2 = \frac{c_0}{d(2d+1)}$, $c_0>0$ and $d\in(0,1/2)$.
\end{lemma}

In the proof of \thref{lemma_konv_alternative} we apply the empirical process non-central limit theorem  of \citet{Dehling.1989}, which uses the Hermite expansion of $1_{\{G(\xi)\leq x\}}-F(x)$. Before proceeding to the proof, we will have a brief look at this concept.\par \bigskip

\textbf{Hermite expansion:} Since function $g(\xi)=1_{\{G(\xi)\leq x\}}-F(x)$ is a measurable function with $\E g(\xi)=0$ and $\E g^2(\xi)<\infty$, $\xi\sim N(0,1)$, i.e. $g\in L^2(\mathbb{R},N)$, we could represent $g$ by its Hermite expansion
\[
g(\xi) = \sum_{i=1}^{\infty}\frac{J_k(x)}{k!}H_k(\xi),
\]
where the equality means convergence in the $L^2$ sense. The $k$-th order Hermite polynomial is given by
\[
H_k(\xi) = (-1)^{k}e^{\xi^2/2}\frac{d^k}{d\xi^k}e^{-\xi^2/2},
\] 
and the coefficients are given by $J_k(x) = \E (1_{\{G(\xi)\leq x\}}H_k(\xi))$, with $J_1(x)=\E(\xi_1 1_{\{\xi_1\leq x\}})=-\varphi(x)$, where $\varphi(x)$ denotes the standard normal density function.
The Hermite rank is defined as $m = \min\{k\geq 0: J_k\neq 0\}$,
the smallest $k$ for which the term in the Hermite expansion is not zero. Since $J_1(x)\neq 0$ for some $x\in \mathbb{R}$, we have Hermite rank $m=1$.\par \bigskip

\textbf{Hermite process:} The limit process $Z_m(t)$ in Theorem 1.1 of \citet{Dehling.1989} is called $m$-th order Hermite process and is defined e.g. in \citet{Taqqu.1978}. If $m=1$, $Z_1(t)$ is the standard Gaussian fractional Brownian motion.\par \bigskip

\begin{proof}[Proof of \thref{lemma_konv_alternative}]
\citet{Dehling.2013} have shown in their Theorem 1 that 
\begin{multline*}
\bigg(\frac{1}{n^{d+3/2}c_d}\sum_{i=1}^{[ns]}\sum_{j=[ns]+1}^n (1_{\{X_i\leq X_j\}}-1/2)\bigg)_{0\leq s \leq 1}\\ \xrightarrow{d} \bigg(\frac{1}{m!}(Z_m(s)-s Z_m(1))\int_{\mathbb{R}}J_m(x)dF(x)\bigg)_{0\leq s \leq 1}
\end{multline*}
for $X_i=G(\xi_i)$, where $G:\mathbb{R}\rightarrow\mathbb{R}$ is a measurable function (that might not be strictly monotone), $F$ is the continuous distribution of $X_i$, $m$ is the Hermite rank of the class functions $1_{\{G(\xi_i)\leq x\}}-F(x)$, and $J_m(x)$, $H_m$ and $(Z_m(s))_{s\in[0,1]}$ are given above. \par \bigskip
Following the proof of Theorem 1 of \citet{Dehling.2013} we will show 
\begin{multline}\label{zz_th_1}
\bigg(\frac{1}{n^{d+3/2}c_d}\sum_{i=1}^{[ns]}\sum_{j=[nt]+1}^n (1_{\{X_i\leq X_j\}}-1/2)\bigg)_{0\leq s \leq t\leq 1}\\ \xrightarrow{d} \bigg(\frac{1}{m!}\big((1-t)Z_m(s)-s (Z_m(1)-Z_m(t)\big)\int_{\mathbb{R}}J_m(x)dF(x)\bigg)_{0\leq s \leq t\leq 1}.
\end{multline}
Since $F$ is a continuous distribution function, $\int_{\mathbb{R}}F(x)dF(x)=1/2$. Denote $F_k(x)=\frac{1}{k}\sum_{i=1}^k1_{\{X_i\leq x\}}$ and $F_{k+1,n}(x)=\frac{1}{n-k}\sum_{i=k+1}^n1_{\{X_i\leq x\}}$. 
Then, 
\begin{align*}
\sum_{i=1}^{[ns]}\sum_{j=[nt]+1}^n (1_{\{X_i\leq X_j\}}-1/2) 
&=[ns](n-[nt])\Big( \int_{\mathbb{R}}\big(F_{[ns]}(x)-F(x)\big)dF_{[nt]+1,n}(x)\Big)\\
& + [ns](n-[nt])\Big( \int_{\mathbb{R}}F(x)d\big(F_{[nt]+1,n}-F\big)(x) \Big).
\end{align*}
Integration by parts yields,
\begin{equation*}
\int_{\mathbb{R}}F(x)d\big(F_{[nt]+1,n}-F\big)(x)  = - \int_{\mathbb{R}}\big(F_{[nt]+1,n}-F\big)(x)dF(x).
\end{equation*}
Hence,
\begin{multline*}
\sum_{i=1}^{[ns]}\sum_{j=[nt]+1}^n (1_{\{X_i\leq X_j\}}-1/2)= [ns](n-[nt])\int_{\mathbb{R}}(F_{[ns]}(x)-F(x))dF_{[nt]+1,n}(x)\\ -[ns](n-[nt])\int_{\mathbb{R}}(F_{[nt]+1,n}(x)-F(x))dF(x).
\end{multline*} 
With the same argument as used in \citet{Dehling.2013}, we show that
\begin{align*}
\bigg(\frac{[ns](n-[nt])}{n^{d+3/2}c_d}\int_{\mathbb{R}}(F_{[ns]}(x)-F(x))dF_{[nt]+1,n}(x)\bigg)_{0\leq s \leq t\leq 1}\\ \xrightarrow{d} \bigg(\frac{(1-t)}{m!}\int_{\mathbb{R}}J_m(x)Z_m(s)dF(x)\bigg)_{0\leq s \leq t\leq 1},
\end{align*}
and
\begin{align*}
\bigg(\frac{[ns](n-[nt])}{n^{d+3/2}c_d}\int_{\mathbb{R}}(F_{[nt]+1,n}(x)-F(x))dF(x)\bigg)_{0\leq s \leq t\leq 1}\\ \xrightarrow{d} \bigg(\frac{s}{m!}\int_{\mathbb{R}}J_m(x)(Z_m(1)-Z_m(t))dF(x)\bigg)_{0\leq s \leq t\leq 1}.
\end{align*}
We do this by applying the Skorohod-Dudley-Wichura representation which yields almost sure convergence, i.e.
\begin{align}
\frac{[ns](n-[nt])}{n^{d+3/2}c_d}&\int_{\mathbb{R}}(F_{[ns]}(x)-F(x))dF_{[nt]+1,n}(x) - \frac{(1-t)}{m!}\int_{\mathbb{R}}J_m(x)Z_m(s)dF(x) \rightarrow 0\label{as_conv1}\\
\frac{[ns](n-[nt])}{n^{d+3/2}c_d}&\int_{\mathbb{R}}(F_{[nt]+1,n}(x)-F(x))dF(x) - \frac{s}{m!}\int_{\mathbb{R}}J_m(x)(Z_m(1)-Z_m(t))dF(x)\rightarrow 0,\label{as_conv2}
\end{align}
almost surely, uniformly in $0<s\leq t <1$. \par \bigskip
Let us start with (\ref{as_conv1}). We can write
\begin{align}
&\frac{[ns](n-[nt])}{n^{d+3/2}c_d}\int_{\mathbb{R}}(F_{[ns]}(x)-F(x))dF_{[nt]+1,n}(x) - \frac{(1-t)}{m!}\int_{\mathbb{R}}J_m(x)Z_m(s)dF(x)\nonumber\\
 &= \frac{(n-[nt])}{n}\int_{\mathbb{R}}\frac{[ns]}{n^{d+1/2}c_d}(F_{[ns]}(x)-F(x))dF_{[nt]+1,n}(x)- (1-t)\int_{\mathbb{R}}J_m(x)\frac{Z_m(s)}{m!}dF(x)\nonumber\\
& = \frac{(n-[nt])}{n}\int_{\mathbb{R}}\Big(\frac{[ns]}{n^{d+1/2}c_d}(F_{[ns]}(x)-F(x))-J_m(x)\frac{Z_m(s)}{m!}\Big)dF_{[nt]+1,n}(x)\nonumber\\
&+  \frac{(n-[nt])}{n}\int_{\mathbb{R}}J_m(x)\frac{Z_m(s)}{m!}d\big(F_{[nt]+1,n}-F\big)(x)\nonumber\\
& + \Big(\frac{(n-[nt])}{n} - (1-t) \Big)\int_{\mathbb{R}}J_m(x)\frac{Z_m(s)}{m!}dF(x).\label{rs_as_conv1}
\end{align}
The empirical process non-central limit theorem of \citet{Dehling.1989} yields
\[
\Big(d_n^{-1}[ns]\big(F_{[ns]}(x)-F(x)\big) \Big)_{x\in[-\infty,\infty],s\in[0,1]} \xrightarrow{d} \Big(J(x)Z(s)\Big)_{x\in[-\infty,\infty],s\in[0,1]},
\]
where $J(x)=J_m(x)$, $Z(x)=Z_m(x)/m!$ and $d_n^2\sim n^{2d+1}c_d^2$.\par 
\citet{Dehling.2013} argue that applying the Skorohod-Dudley-Wichura representation yields almost sure convergence, i.e.
\begin{equation}\label{as_conv_sdw}
\sup_{s,x}\big|d_n^{-1}[ns]\big(F_{[ns]}(x)-F(x)\big)-J(x)Z(x)\big| \rightarrow 0\qquad \text{a.s.}
\end{equation}
Thus, the first term on the right-hand side of (\ref{rs_as_conv1}) converges to 0 almost surely, uniformly in $0<s\leq t <1$.\par 
Furthermore, we note that 
\begin{align*}
\frac{(n-[nt])}{n}&\int_{\mathbb{R}}J(x)Z(s)d\big(F_{[nt]+1,n}-F\big)(x)\\
&= Z(s)\Big[ \frac{(n-[nt])}{n}\int_{\mathbb{R}}J(x)dF_{[nt]+1,n}(x)-\frac{(n-[nt])}{n}\int_{\mathbb{R}}J(x)dF(x) \Big]\\
& =  Z(s)\Big[\frac{1}{n}\sum_{i=[nt]+1}^{n}J(X_i)- \frac{(n-[nt])}{n}\E(J(X_i))\Big]\\
&= Z(s) \frac{1}{n}\sum_{i=1}^{n} \big( J(X_i)-\E(J(X_i))\big) -Z(s) \frac{1}{n}\sum_{i=1}^{[nt]} \big( J(X_i)-\E(J(X_i))\big).
\end{align*}
Note that $(J(X_i))$ is ergodic since the process $(X_i)$ is ergodic and $J$ is a measurable function. By the ergodic theorem, $\frac{1}{n}\sum_{i=1}^{n} \big( J(X_i)-\E(J(X_i))\big) \rightarrow 0$ almost surely. This implies that $\sum_{i=1}^{n} \big( J(X_i)-\E(J(X_i))\big)= o(n)$ and hence
\[
\max_{0\leq k \leq n}\Big|\sum_{i=1}^{k} \big( J(X_i)-\E(J(X_i))\big)\Big|=o(n)
\]
almost surely as $n\rightarrow \infty$. Thus, $\frac{1}{n}\sum_{i=1}^{[nt]} \big( J(X_i)-\E(J(X_i))\big) \rightarrow 0$ almost surely for all $0\leq t \leq 1$. Therefore, the second term on the right-hand side of (\ref{rs_as_conv1}) converges to 0 almost surely, uniformly in $0<s\leq t <1$.\par
Also the third term on the right-hand side of (\ref{rs_as_conv1}) converges to 0, since, as $n\rightarrow\infty$, $
\big((n-[nt])/n - (1-t)\big) \rightarrow 0$, and $\int_{\mathbb{R}}J_m(x)\frac{Z_m(s)}{m!}dF(x)$ is bounded. This finishes the proof of (\ref{as_conv1}).\par 
Note that
\[
F_{[nt]+1,n}(x) = \frac{n}{n-[nt]}F_{n}(x)-\frac{[nt]}{n-[nt]}F_{[nt]}(x),
\]
and hence,
\[
(n-[nt])\big(F_{[nt]+1,n}(x)-F(x) \big) = n\big(F_n(x)-F(x) \big)-[nt]\big(F_{[nt]}(x)-F(x)\big).
\]
Then the proof of (\ref{as_conv2}) follows using again (\ref{as_conv_sdw}). Thus, (\ref{zz_th_1}) is shown.\par 
Note that this result holds for $X_i=G(\xi_i)$, but in our lemma we consider $X_i=\xi_i$, where $(\xi_j)$ is a stationary mean zero Gaussian process with auto-covariances $\gamma_k\sim k^{2d-1}c_0$, $d\in(0,1/2)$. In this case, $J_1(x)=-\varphi(x)$, where $\varphi(x)$ denotes the standard normal density function and $\int_{\mathbb{R}}J_1(x)dF(x)=-\frac{1}{2\sqrt{\pi}}$, since $F$ is the normal distribution function. Furthermore, $J_1(x)\neq 0$ for all $x$ and hence, we have Hermite rank $m=1$. Therefore, $(Z_1(s))$ denotes the standard fractional Brownian motion process $(W_H(s))$. Thus, the limit in (\ref{zz_th_1}) equals
\[
\frac{1}{2\sqrt{\pi}}\Big(s(W_H(1)-W_H(t))-(1-t)W_H(s)\Big),
\]
which proves the lemma.
\end{proof}
\par\bigskip
\begin{lemma}\thlabel{lemma_ld_alternative}
Assume that the assumptions of \thref{theorem_alternative} hold. Then,
\begin{align*}
&\bigg[\frac{1}{n^{d+3/2}c_d}\max_{1\leq k\leq \hat{k}}\Big|\sum_{i=1}^{k}\sum_{j=k+1}^{\hat{k}}(1_{\{\xi_i\leq \xi_j\}}-1/2)\Big|, \frac{1}{n^{d+3/2}c_d}\max_{\hat{k}< k\leq n}\Big|\sum_{i=\hat{k}+1}^{k}\sum_{j=k+1}^{n}(1_{\{\xi_i\leq \xi_j\}}-1/2)\Big| \bigg]\\
&\overset{d}{\longrightarrow} \bigg[\frac{\zeta}{2\sqrt{\pi}}\sup_{0\leq t\leq \zeta}\big|W_H(t)-\frac{t}{\zeta}W_H(\zeta)\big|,\\
&\qquad\qquad\qquad\qquad\qquad\qquad\frac{1-\zeta}{2\sqrt{\pi}}\sup_{\zeta\leq t\leq 1}\big|W_H(t)-W_H(\zeta)-\frac{t-\zeta}{1-\zeta}(W_H(1)-W_H(\zeta))\big| \bigg],
\end{align*}
where $c_d^2 = \frac{c_0}{d(2d+1)}$, $c_0>0$, $d\in(0,1/2)$, $W_H$ is a standard fractional Brownian motion, $H=d+1/2$ and
\begin{equation}\label{xi}
\zeta = \inf\Big\{ t\geq 0: \sup_{0\leq s\leq 1}|W_H(s)-s W_H(1)|=|W_H(t)-t W_H(1)|\Big\}.
\end{equation}
\end{lemma}

\begin{proof}
Denote for $0\leq s\leq t \leq 1$
\begin{align*}
\tilde{U}_{n}(s,t) &=  \frac{1}{n^{d+3/2}c_d}\sum_{i=1}^{[ns]}\sum_{j=[nt]+1}^n (1_{\{\xi_i\leq \xi_j\}}-1/2),\\
 \tilde{W}_H(s,t) &= -\frac{1}{2\sqrt{\pi}}\big((1-t)W_H(s)-s(W_H(1)-W_H(t))\big)
\end{align*}
and note that by \thref{lemma_konv_alternative}, $(\tilde{U}_{n}(s,t))_{s,t}\xrightarrow{d}(\tilde{W}_{H}(s,t))_{s,t}$. Furthermore, we denote 
\begin{align*}
\tilde{U}_{n,1}(t) &= \frac{1}{n^{d+3/2}c_d}\max_{1\leq k\leq nt}\Big|\sum_{i=1}^{k}\sum_{j=k+1}^{nt}(1_{\{\xi_i\leq \xi_j\}}-1/2)\Big|,\\
\tilde{U}_{n,2}(t) &=\frac{1}{n^{d+3/2}c_d}\max_{nt< k\leq n}\Big|\sum_{i=nt+1}^{k}\sum_{j=k+1}^{n}(1_{\{\xi_i\leq \xi_j\}}-1/2)\Big|,\\
\tilde{W}_{H,1}(t) &= \frac{t}{2\sqrt{\pi}}\sup_{0\leq s\leq t}\big|W_H(s)-\frac{s}{t}W_H(t)\big|,\\
\tilde{W}_{H,2}(t) &= \frac{1-t}{2\sqrt{\pi}}\sup_{t\leq s\leq 1}\big|(W_H(s)-W_H(t))-\frac{1-s}{1-t}(W_H(1)-W_H(t))\big|.
\end{align*}
Since 
\[
\sum_{i=1}^{k}\sum_{j=k+1}^{nt}(1_{\{\xi_i\leq \xi_j\}}-1/2) = \sum_{i=1}^{k}\sum_{j=k+1}^{n}(1_{\{\xi_i\leq \xi_j\}}-1/2)-\sum_{i=1}^{k}\sum_{j=nt+1}^{n}(1_{\{\xi_i\leq \xi_j\}}-1/2),
\]
we can write $\tilde{U}_{n,1}(t)=\sup_{0\leq s\leq t}|\tilde{U}_{n}(s,s)-\tilde{U}_{n}(s,t)|$ and with a similar argument $\tilde{U}_{n,2}(t)=\sup_{t\leq s\leq 1}|\tilde{U}_{n}(s,s)-\tilde{U}_{n}(t,s)|$. Note that $\tilde{W}_{H,1}(t)=\sup_{0\leq s\leq t}|\tilde{W}_{H}(s,s)-\tilde{W}_{H}(s,t)|$ and $\tilde{W}_{H,2}(t)=\sup_{t\leq s\leq 1}|\tilde{W}_{H}(s,s)-\tilde{W}_{H}(t,s)|$. Thus, the same continuous mapping transforms $\tilde{U}_{n}(s,t)$ into the vector $(\hat{k}/n,\tilde{U}_{n,1}(t),\tilde{U}_{n,2}(t))$ and $\tilde{W}_H(s,t)$ into $(\zeta,\tilde{W}_{H,1}(t),\tilde{W}_{H,2}(t))$, where $\zeta$ is given in (\ref{xi}). Hence, by the continuous mapping theorem and \thref{lemma_konv_alternative} 
\[
\Big(\hat{k}/n,\tilde{U}_{n,1}(t),\tilde{U}_{n,2}(t)\Big)\xrightarrow{d}\Big(\zeta,\tilde{W}_{H,1}(t),\tilde{W}_{H,2}(t)\Big).
\]
Applying the mapping $(z,x(t),y(t))\mapsto (x(z),y(z))$ to both vectors finishes the proof.
\end{proof}
\par

\begin{proof}[\textbf{Proof of \thref{theorem_alternative}}]\par 
By \thref{lemma_ld_alternative},
\begin{align*}
T_{n,1} &= \hat{k}^{-3/2}\max_{1\leq k\leq \hat{k}}\Big|\sum_{i=1}^{k}\sum_{j=k+1}^{\hat{k}}(1_{\{\xi_i\leq \xi_j\}}-1/2)\Big|\\
&= \frac{n^{d+3/2}c_d}{\hat{k}^{3/2}}\frac{1}{n^{d+3/2}c_d}\max_{1\leq k\leq \hat{k}}\Big|\sum_{i=1}^{k}\sum_{j=k+1}^{\hat{k}}(1_{\{\xi_i\leq \xi_j\}}-1/2)\Big|= \frac{n^{d+3/2}c_d}{\hat{k}^{3/2}} O_P(1).
\end{align*}
Similar argument yields $T_{n,2}= \frac{n^{d+3/2}c_d}{(n-\hat{k})^{3/2}}O_P(1)$. Thus, to prove \thref{theorem_alternative} it remains to show $\frac{n^{d+3/2}c_d}{\hat{k}^{3/2}}\rightarrow_p\infty$ and $\frac{n^{d+3/2}c_d}{(n-\hat{k})^{3/2}}\rightarrow_p\infty$. The proof of \thref{lemma_ld_alternative} yields $\hat{k}/n\xrightarrow{d} \zeta$, where $\zeta$ is given in (\ref{xi}), and hence, $(\hat{k}/n)^{3/2}$ and $((n-\hat{k})/n)^{3/2}$ are asymptotically bounded away from zero. Since $d>0$, $n^{d}\rightarrow\infty$ as $n\rightarrow\infty$. Thus, $T_{n,1}\rightarrow_p\infty$ and $T_{n,2}\rightarrow_p\infty$. This finishes the proof of \thref{theorem_alternative}.
\end{proof}

\section*{Data Availability}
The data used in Section 4 are property of the German state Saxony and are therefore not openly available but can be requested by the Saxon State Office for Environment, Agriculture and Geology.

\section*{Acknowledgement}
The author would like to thank Herold Dehling, Liudas Giraitis and Isabel Garcia for valuable discussions. The research was supported by the Collaborative Research Centre 823 \emph{Statistical modelling of nonlinear dynamic processes} and the Konrad-Adenauer-Stiftung. The author thanks Svenja Fischer for providing the hydrologic data set.

\

%\nocite{*}
\footnotesize
%volume nutzen
%article
%\bibitem [] {}  (). . \emph{} \textbf{} .
%book
%\bibitem [Billingsley(1968)] {Billingsley.1968} Billingsley, P. (1968). \emph{Convergence of Probability Measures}. Wiley, New York.

\end{document}